\documentclass[aps,pre,reprint,superscriptaddress,floatfix,amsmath,amssymb,nofootinbib]{revtex4-2}

\usepackage{dcolumn}
\usepackage{bm}
\usepackage{float}
\usepackage[utf8]{inputenc}
\usepackage{amsthm,amssymb,nicefrac,pifont}
\usepackage[table, dvipsnames]{xcolor}
\definecolor{syn}{HTML}{85B6B2}
\definecolor{redu}{HTML}{4E79A7}
\definecolor{unq1}{HTML}{E49444}
\definecolor{unq2}{HTML}{E15759}

\newtheorem{definition}{Definition}
\theoremstyle{definition}
\newtheorem{example}{Example}
\newtheorem{theorem}{Theorem}
\newtheorem{corollary}{Corollary}

\usepackage{bold-extra}
\usepackage{booktabs}
\usepackage{longtable}
\usepackage{array, multirow,url}
\usepackage[bookmarks=false,colorlinks=true,urlcolor=blue,linkcolor=blue,citecolor=blue]{hyperref}
\usepackage[ruled]{algorithm2e}
\usepackage{subfigure}
\usepackage{physics}
\usepackage{dsfont}
\usepackage{csquotes}
\usepackage[noabbrev,capitalise]{cleveref}
\renewcommand\autoref[1]{\cref{#1}}
\usepackage{comment}
\usepackage{siunitx}
\usepackage{listings}
\usepackage{tikz}
\usetikzlibrary{backgrounds, matrix,fit,matrix,decorations.pathreplacing, calc, positioning, patterns}
\usepackage{graphicx}

\crefname{theorem}{thrm.}{thrms.}
\crefname{enumi}{case}{cases}
\crefname{equations}{Equations}{Equations}

\definecolor{codegreen}{rgb}{0,0.6,0}
\definecolor{codegray}{rgb}{0.5,0.5,0.5}
\definecolor{codepurple}{rgb}{0.58,0,0.82}
\definecolor{backcolour}{rgb}{0.95,0.95,0.92}
\lstdefinestyle{mystyle}{
    backgroundcolor=\color{backcolour},   
    commentstyle=\color{codegreen},
    keywordstyle=\color{magenta},
    numberstyle=\tiny\color{codegray},
    stringstyle=\color{codepurple},
    basicstyle=\ttfamily\footnotesize,
    breakatwhitespace=false,         
    breaklines=true,                 
    captionpos=b,                    
    keepspaces=true,                 
    numbers=left,                    
    numbersep=5pt,                  
    showspaces=false,                
    showstringspaces=false,
    showtabs=false,                  
    tabsize=2
}

\DeclareFontFamily{U}{mathx}{\hyphenchar\font45}
\DeclareFontShape{U}{mathx}{m}{n}{
      <5> <6> <7> <8> <9> <10>
      <10.95> <12> <14.4> <17.28> <20.74> <24.88>
      mathx10
      }{}
\DeclareSymbolFont{mathx}{U}{mathx}{m}{n}
\DeclareMathSymbol{\bigtimes}{1}{mathx}{"91}

\begin{document}

\preprint{APS/123-QED}

\title{Partial Information Decomposition for Continuous Variables based on Shared Exclusions:  Analytical Formulation and Estimation }

\author{David A. Ehrlich}
\thanks{Equally contributing first author}
\email{davidalexander.ehrlich@uni-goettingen.de}
\affiliation{%
 Göttingen Campus Institute for Dynamics of Biological Networks\\
 Universität Göttingen\\
 Göttingen
}

\author{Kyle Schick-Poland}%
\thanks{Equally contributing first author}
\email{kyle.schick-poland@uni-goettingen.de}
\affiliation{%
 Göttingen Campus Institute for Dynamics of Biological Networks\\
 Universität Göttingen\\
 Göttingen
}
\affiliation{%
 Honda Research Institute Europe GmbH, Offenbach am Main
}

\author{Abdullah Makkeh}
\affiliation{%
 Göttingen Campus Institute for Dynamics of Biological Networks\\
 Universität Göttingen\\
 Göttingen
}
\author{Felix Lanfermann}
\author{Patricia Wollstadt}
\affiliation{%
 Honda Research Institute Europe GmbH, Offenbach am Main
}
\author{Michael Wibral}
\affiliation{%
 Göttingen Campus Institute for Dynamics of Biological Networks\\
 Universität Göttingen\\
 Göttingen
}
\date{\today}

\begin{abstract}
Describing statistical dependencies is foundational to empirical scientific research. For uncovering intricate and possibly non-linear dependencies between a single target variable and several source variables within a system, a principled and versatile framework can be found in the theory of Partial Information Decomposition (PID). Nevertheless, the majority of existing PID measures are restricted to categorical variables, while many systems of interest in science are continuous. In this paper, we present a novel analytic formulation for continuous redundancy--a generalization of mutual information--drawing inspiration from the concept of shared exclusions in probability space as in the discrete PID definition of $I^\mathrm{sx}_\cap$. Furthermore, we introduce a nearest-neighbor based estimator for continuous PID, and showcase its effectiveness by applying it to a simulated energy management system provided by the Honda Research Institute Europe GmbH. This work bridges the gap between the measure-theoretically postulated existence proofs for a continuous $I^\mathrm{sx}_\cap$ and its practical application to real-world scientific problems.
\end{abstract}

\maketitle

\section{Introduction}
The pursuit of discovering and quantifying dependencies between different experimental variables lies at the heart of empirical and data-driven scientific research. However, conventional tools such as correlation analysis might fail to capture relevant associations if these associations are non-linear. This ascertains the need for a comprehensive framework capable of capturing both linear and non-linear dependencies. Such a framework can be found in information theory, which was originally introduced for the analysis of communication channels by Claude Shannon~\cite{shannon} and has since become a general-purpose approach to data analysis. 

A basic quantity of information theory is the mutual information~\cite{shannon,fano1961transmission}

\begin{equation*}
    I(T:S) = \sum_{t, s} p_{T,S}(t, s)\log_2 \frac{p_T(t) \, p_S(s)}{p_{T,S}(t,s)} \, ,
\end{equation*}
which captures the dependency between two variables $T$ and $S$ by quantifying how much the uncertainty about one variable can be reduced by observing the respective other, and which is used as a general measure of dependency between two variables in many applications \cite{coverthomas}.

In many situations, however, there are variables of interest $T$ which depend not only on one but multiple variables $\bm S = (S_1, \dots S_n)$, and where the distinction between the $S_i$ is conceptually important. Hence arises the need to unravel the specific contributions of these source variables $S_i$ to the information contained in the target $T$. Note that this information may be distributed among the source variables in very different ways: While parts of the information might exclusively be available from one variable but not from others (\enquote{unique information}), other parts can be obtained from either one or another (\enquote{redundant information}) and finally some pieces of information might only be revealed when considering multiple sources simultaneously (\enquote{synergistic information}). Identifying and quantifying these information \emph{atoms} is the subject of Partial Information Decomposition (PID)~\cite{williams2010, gutknecht2021}, which has been gaining popularity as a tool for the detailed information-theoretic analysis of variable dependencies, for example in neuroscience \cite{wibral2017,wollstadt2022rathbun,sherrill2021partial,pica2017quantifying,faes2017multiscale,schulz2021gaba,luppi2020synergistic}, machine learning \cite{tax2017partial,ehrlich2023complexity,wollstadt2023feature,graetz2023infomorphic,tokui2021disentanglement,liang2023multimodal,liang2023factorized,ingel2022quantifying}, engineering \cite{wollstadt2021turbofan}, sociology~\cite{varley2022}, linguistics~\cite{socolof2022measuring} and climatology~\cite{glowienka2020comparing}.

The PID framework was originally envisioned by Williams and Beer~\cite{williams2010}, who showed that the information atoms of unique, redundant, and synergistic information could not be defined using classical information-theoretic terms, but required the introduction of novel axioms. Based on their proposed set of axioms, the authors introduced a measure of redundant information as a generalization of mutual information, that allowed to quantify the PID atoms. Following the original work of Williams and Beer, a series of different quantification schemes for PIDs has been proposed, which mostly adopt the proposed overall structure and either adopt the original set of axioms or modify it~\cite{broja,makkeh,finn2018pointwise,ince,kolchinsky2022novel}. In particular, most works propose alternative measures of redundancy, which are grounded in a multitude of different and partly non-compatible desiderata~\cite{makkeh,broja}.

However, while most variables of interest encountered in science and engineering are continuous-valued, most of the PID measures suggested to this date are only valid for categorical random variables, i.e., variables with a discrete alphabet. Only a few continuous PID measures have been introduced in recent time, which each have distinct  operational interpretations: \citet{barrett2015exploration} were among the first to introduce a fully continuous PID for Gaussian random variables. Drawing on concepts of game theory, the measure by \citet{ince} based on common changes in surprisal can be straightforwardly applied to continuous variables. Similarly, the measure by \citet{kolchinsky2022novel}, which is build on the notion of the Blackwell order, and the redundant information neural estimation by \citet{kleinman2021redundant} also transfer canonically to the continuous domain. Finally, the measure introduced by \citet{pakman} provides a continuous PID definition based on the decision-theoretic discrete concepts introduced by \citet{broja}, while \citet{milzman2022signed} introduce continuous generalizations of $I^\mathrm{min}_\cap$~\cite{williams2010} and $I^\mathrm{PM}_\cap$~\cite{finn2018pointwise}.

Overall, these measures all draw on different desiderata which make the measures applicable to different situations that fit their operational interpretation. If the system is best described by a memoryless agent acting on an ensemble of states, the shared-exclusion PID measure $I^\mathrm{sx}_\cap$ is most suitable, which has been introduced by \citet{makkeh} for categorical variables. Recently, \citet{schick-poland} showed that an extension of this measure to continuous variables is possible, by proving that a continuous version of the  measure is measure-theoretically well-defined. However, since the provided existence proofs are not fully constructive, neither an analytical definition nor a practical estimation procedure for such a continuous $I^\mathrm{sx}_\cap$ has been suggested to this date. With this paper, we fill the gap for a continuous PID measure based on the shared-exclusion principle which draws only on concepts of probability and information theory.

The main contributions of this paper are (1) the introduction of a tractable analytical PID definition inspired by the measure-theoretic definition of continuous $I^\mathrm{sx}_\cap$ from \citet{schick-poland} and its application to simple theoretical examples in \autoref{sec:analytical}, (2) the development of an estimator for the associated redundancy measure, which draws on concepts of the k-nearest-neighbours based estimator for mutual information by \citet{kraskov} in \autoref{sec:estimator} and (3) the demonstration of the efficacy of our continuous $I^\mathrm{sx}_\cap$ measure in uncovering variable dependencies in data from an energy management system in \autoref{sec:application}.

\section{The PID problem for two source variables}
\label{sec:pid_bivariate}

\begin{figure}
    \centering
    \includegraphics[scale=1]{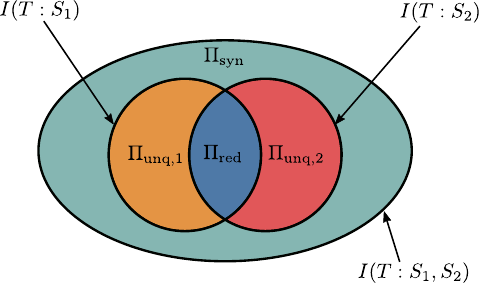}
    \caption{\textbf{Partial Information Decomposition reveals the intricate interdependencies between classical information-theoretic quantities~\cite{williams2010}.} Using PID, the mutual information $I(T:S_1, S_2)$ (large oval) and the marginal information terms $I(T:S_1)$ and $I(T:S_2)$ (circles) can simultaneously be dissected into a total of four information atoms (colored areas): The unique information that $S_1$ carries about $T$, $\Pi_\mathrm{unq,1}$ (orange), and likewise the unique information of $S_2$, $\Pi_\mathrm{unq,2}$ (red); the redundant information of $S_1$ and $S_2$ about $T$, $\Pi_\mathrm{red}$ (blue), and the synergistic information exclusively contained in $S_1$, and $S_2$, $\Pi_\mathrm{syn}$ (teal). The marginal mutual information terms $I(T:S_i)$ can be constructed from the \emph{redundant} and the corresponding \emph{unique} atoms, while the joint mutual information $I(T:S_1, S_2)$ contains of all four atoms.}
    \label{fig:pid}
\end{figure}

To provide a succinct account of the PID problem and its intuitive meaning, we start with a discussion of PID using only two source variables---referred to as \emph{bivariate} PID---in this section. A comprehensive treatment of the general multivariate PID is deferred to  \cref{sec:nd_pid}.

The mutual information $I(T:\bm S)$ that two source variables $\bm S=(S_1, S_2)$ hold about a target variable $T$ can be conveyed by the source variables in four distinct ways (see \cref{fig:pid}): Some parts of the information might be \emph{unique} to either $S_1$ or $S_2$, meaning they are inaccessible from the other variable. Other parts may be \emph{redundant} between both, meaning that they can be obtained from either source, while finally some parts only become accessible \emph{synergistically} when both sources are observed together---like a cipher text and its corresponding password are jointly necessary to recover the plain text. We denote these PID atoms by $\Pi_\mathrm{unq,1}$ and $\Pi_\mathrm{unq,2}$ for the unique information, by $\Pi_\mathrm{red}$ for the redundant information, and by $\Pi_\mathrm{syn}$ for the synergistic information.

To understand how the joint mutual information $I(T: \bm S)$ relates to the PID atoms, a first step is to look at the \emph{marginal} mutual information terms $I(T:S_1)$ and $I(T:S_2)$, which quantify the statistical dependence between the target $T$ and a single source variable. These two information quantities, however, need not be disjoint, because each mutual information, while containing the respective unique atoms $\Pi_\mathrm{unq,1}$ or $\Pi_\mathrm{unq,2}$, also contain the \emph{redundant} atom $\Pi_\mathrm{red}$. Finally, besides the two unique and the redundant term, the joint mutual information $I(T: \bm S)$ also contains the synergistic atom $\Pi_\mathrm{syn}$, which is part of neither of the marginal mutual information terms because it is inaccessible from any single source alone. These requirements lead to the so-called \emph{consistency equations} for two source variables~\cite{williams2010,gutknecht2021}:
\begin{equation}
\begin{aligned}
    &I(T:S_1) &&= \Pi_\mathrm{red} + \Pi_\mathrm{unq,1}\\
    &I(T:S_2) &&= \Pi_\mathrm{red} + \Pi_\mathrm{unq,2}\\
    &I(T:\bm{S}) &&= \Pi_\mathrm{red} + \Pi_\mathrm{unq,1} + \Pi_\mathrm{unq,2} + \Pi_\mathrm{syn} \, .
    \label[equations]{eq:consistency}
\end{aligned}
\end{equation}

Note that, in the bivariate case, there are four unknown PID atoms, yet only three consistency equations providing constraints. Hence, the system is underdetermined and has to be resolved by defining an additional quantity. Typically, this is done by explicitly quantifying the redundancy $\Pi_\mathrm{red}$ (e.g., \cite{williams2010,finn2018pointwise,makkeh}). For more than two source variables, a generalized notion of redundancy $I_\cap$ needs to be introduced which captures the redundancy between \emph{sets} of source variables (see \cref{sec:nd_pid_lattice}).

In recent years, a variety of different redundancy measures has been suggested which fulfill a number of partly incompatible assumptions and have different operationalizations, e.g.~from decision~\cite{broja} or game theory~\cite{ince}. In this work, we focus on the $I^\mathrm{sx}_\cap$ redundancy measure~\cite{makkeh}, which is rooted in the probability- and information-theoretic considerations of Fano~\cite{fano1961transmission}.

\subsection{The $I^{\mathrm{sx}}_\cap$ redundancy measure and the measure-theoretic definition of continuous $I^{\mathrm{sx}}_\cap$}

The $I^\mathrm{sx}_\cap$ redundancy measure has originally been defined for discrete random variables~\cite{makkeh}. To understand the intuition behind the definition, consider a discrete target random variable $T$ and two discrete source random variables $S_1$ and $S_2$, taking on realizations $t$ and $(s_1, s_2)$, respectively. 
Notice that the joint mutual information $I(T:S_1, S_2)$ between the target and both sources can be written as an expectation value of \emph{local} mutual information terms $i(t:s_1, s_2)$ as~\cite{coverthomas}
\begin{align}
    \label{eq:sum_mi}
    I(T:S_1, S_2) = \sum_{t, s_1, s_2} p_{T,S}(t,s_1,s_2) \, i(t:s_1, s_2) \, ,
\end{align}
where
\begin{subequations}
\begin{align}
    i(t: s_1, s_2) 
    &= \log_2\left[\frac{p_{T|S}(t |s_1, s_2)}{p_T(t)}\right]\\
    &= \log_2\left[\frac{p_{T|S}(t | S_1 = s_1 \land S_2 = s_2)}{p_T(t)}\right] \, . \label{eq:local_mi_logicstatement}
\end{align}
\label{eq:local_mi}
\end{subequations}

The local mutual information reflects how we need to update our beliefs about the occurrence of the target event $T=t$ when we know both source events $S_1=s_1$ \emph{and} $S_2=s_2$ have occurred~\cite{fano1961transmission}: If knowledge of the event $S_1=s_1$ and $S_2=s_2$ happening increases the likelihood of guessing the correct target event $T=t$ (i.e.,~$p_{T|S_1,S_2}(t|s_1, s_2) > p_T(t)$), the event is called \emph{informative} and $i(t:s_1, s_2)$ is positive. Vice versa, if one is less likely to guess the target event $T=t$ after knowing the source events (i.e.,~$p_{T|S_1,S_2}(t|s_1, s_2) < p_T(t)$), the event is called \emph{misinformative} and $i(t:s_1, s_2)<0$.

In the second transformation step, i.e., in \cref{eq:local_mi_logicstatement}, we explicitly state the logical statement whose observation leads us to update our beliefs about $T=t$. \citet{makkeh} suggest that building on this basic idea, a measure for redundancy can be constructed by instead measuring how our beliefs about $T=t$ are updated if we only know that \emph{either} $S_1=s_1$ \emph{or} $S_2=s_2$ have occurred, which leads to the local expression
\begin{align}
    \label{eq:discrete_bivariate_local_sxpid}
    i_\cap^\mathrm{sx}(t: s_1; s_2) = \log_2\left[\frac{p_{T \vert S}(t | S_1 = s_1 \lor S_2 = s_2)}{p_T(t)}\right] \, .
\end{align}

Averaging this quantity with weights $p(t, s_1, s_2)$ as in \cref{eq:sum_mi} then yields a \emph{global} measure of redundant information between sources $S_1$ and $S_2$ about $T$, denoted by $I^\mathrm{sx}_\cap$, i.e.,~\cite{makkeh}
\begin{equation}
    I^\mathrm{sx}_\cap = \sum_{s_1,s_2,t} p_{T, S}(t,s_1,s_2) \; i_\cap^\mathrm{sx}(t: s_1; s_2) \, .
    \label{eq:discrete_bivariate_global_sxpid}
\end{equation}

This redundancy measure is symmetric with respect to swapping the sources $S_1$ and $S_2$ and invariant towards any arbitrary bijective relabelling of the discrete realizations of any of the random variables. Furthermore, it fulfills a target chain rule and is differentiable with respect to the underlying probability distribution~\cite{makkeh}.

Equivalently, \cref{eq:discrete_bivariate_local_sxpid} can be derived from the regions of probability space which $S_1=s_1$ and $S_2=s_2$ jointly render impossible (see \cref{appendix:isx_visual} for a more visual explanation), which is why the PID definition is referred to as ``shared exclusion PID'' or ``SxPID'' for short. For more than two source variables, the logical expression in the numerator has to be extended to account for disjunctions between sets of source events occurring simultaneously (see \cref{sec:multivariate_sxpid}).

In previous work, \citet{schick-poland} showed that the $I^\mathrm{sx}_\cap$ redundancy measure can be generalized to continuous source and target variables. In particular, the authors proved the existence of a measure-theoretic generalization of $I^\mathrm{sx}_\cap$, which can handle arbitrarily many source variables while being invariant under isomorphic transformations and differentiable with respect to the joint distribution of sources and target. The proof of the existence of this generalization builds on established theorems from probability and measure theory. It is shown that the derived measure inherits all the above mentioned desirable properties of the discrete $I^\mathrm{sx}_\cap$ measure and is applicable to any possible discrete or continuous variable settings.

However, due to the provided proofs not being entirely constructive, a concrete analytic formulation of this continuous $I^\mathrm{sx}_\cap$ measure remains elusive. While proven to exist from measure-theoretic principles, the way to construct a transformation-invariant shared-exclusion redundancy measure is still not explicit. Therefore, in the next section we provide an analytical definition for a purely continuous $I^\mathrm{sx}_\cap$ measure which, while not adhering to all properties described by \citet{schick-poland}, provides a practical and interpretable shared-exclusion based continuous PID definition. In the following, we show how such a formulation can be obtained, present an estimation procedure and discuss its relation to the measure-theoretic derivation.

\subsection{An analytical formulation of a continuous $I^{\mathrm{sx}}_\cap$ measure}
\label{sec:analytical}

\subsubsection{Definition}

To make a continuous $I^\mathrm{sx}_\cap$ measure available to practitioners, an important first step is to find an analytic formulation. In this section, we present such a formulation and show how it can be computed numerically for a given probability density function. Since in most practical applications one has only access to a finite sample of data from an unknown distribution, we introduce an estimator for this redundancy measure based on k-nearest-neighbour distances in \cref{sec:estimator}.

The main difficulty for finding an analytical formulation for the continuous $I^\mathrm{sx}_\cap$ measure is due to the disjunction (logical \enquote{or}) in \cref{eq:discrete_bivariate_local_sxpid}. While conjunctions (logical \enquote{and}) are ubiquitous in classical continuous information theory and correspond to simple joint probability densities, disjunctive statements have no established counterpart in the continuous regime. In the following we will introduce an intuitive notion of how such a quantity can be defined, while we defer a more rigorous derivation to \cref{appendix:derivation_density_logical_statements}.

For discrete random variables, the probability of the event $S_1 = s_1 \lor S_2 = s_2$ is given by the inclusion-exclusion-rule as
\begin{equation}
\begin{gathered}
    \mathbb{P}(S_1 = s_1 \lor S_2 = s_2)\\
    = p_{S_1}(s_1) + p_{S_2}(s_2) - p_{S_1, S_2}(s_1, s_2) \, ,
\end{gathered}
    \label{eq:probability_masses}
\end{equation}
where the last term  $p_{S_1, S_2}(s_1, s_2)$ is the probability of both $S_1=s_1$ \emph{and} $S_2=s_2$ occurring, which is double-counted in the sum of marginals and thus needs to be subtracted once.

To make the leap from discrete to continuous variables, we follow the steps of Jaynes' derivation of the differential entropy~\cite{jaynes_2003} by reinterpreting the probability mass functions $p_{S_i}$ of \cref{eq:probability_masses} as binned distributions of an underlying continuous probability density $f_{S_i}$. That is, we assume the density $f_{S_i}$ to be given, such that the probability mass of a specific binned realization can be calculated as
\begin{equation*}
    p(s_i) = \int_{U_{s_i}} \mathrm{d}s_i' \, f_{S_i}(s_i') \approx |U_{s_i}| \, f_{S_i}(s_i) \, ,
\end{equation*}
where $\{U_{s_i}\}_{s_i}$ is a partition of the support of $f_{S_i}$ and $|U_{s_i}|$ is the area of one of the sets of the partition. The probability of the binned disjunction, i.e., the probability of either $S_1$ falling into the bin $U_{s_1}$ \emph{or} $S_2$ falling into $U_{s_2}$ then becomes
\begin{equation}
    \begin{gathered}
    \mathbb{P}(S_1 \in U_{s_1} \lor S_2 \in U_{s_2})\\
    \approx |U_{s_1}| \, f_{S_1}(s_1) + |U_{s_2}| \, f_{S_2}(s_2)\\
    - |U_{s_1}||U_{s_2}| \, f_{S_1, S_2}(s_1, s_2) \, .
    \end{gathered}
    \label{eq:probability_masses_approx}
\end{equation}

To get to a fully continuous description, one now has to take the limit of the partitionings $\{U_{s_1}\}_{s_1}$ and $\{U_{s_2}\}_{s_2}$ becoming increasingly fine. If this limit is sufficiently well-behaved, the local density of the evaluation points $s_i$ can in the limit be represented by a function (referred to by Jaynes as an ``invariant measure''~\cite{jaynes1968prior}) $m(s_i)$ such that
\begin{equation*}
    \lim_{n \to \infty} (n |U_{s_i}|) = \frac{1}{m(s_i)} \, ,
\end{equation*}
with
\begin{equation*}
    \lim_{n \to \infty} \frac{1}{n} (\text{number of points in } A) = \int_A \mathrm{d}s_i \, m(s_i)
\end{equation*}
for arbitrary open sets $A$ in the sample space of $S_i$. Specifically, the measure $m_{S_i}$ represents the density of evaluation points, such that a more dense collection of evaluation points causes the corresponding binning intervals to shrink more rapidly. From this, we arrive at the formulation
\begin{equation}
\begin{gathered}
    \mathbb{P}(S_1 \in U_{s_1} \lor S_2 \in U_{s_2}) \\
    \approx \frac{f_{S_1}(s_1)}{n \, m_{S_1}(s_1)} + \frac{f_{S_2}(s_2)}{n \, m_{S_2}(s_2)} - \frac{f_{S_1, S_2}(s_1, s_2)}{n^2 \, m_{S_2}(s_2) m_{S_2}(s_2)}
\end{gathered}
\end{equation}
for the probability mass of a sufficiently small disjunction in the continuous regime. Note that, while the whole expression naturally tends to zero in the limit $n \to \infty$, the subtracted term scales with the number of discrete points squared, $n^2$. Thus, it will become negligible compared to the two terms before under the assumption that $f_{S_1, S_2}(s_1, s_2)$ is finite everywhere.

Nonetheless, neglecting the third term still leaves a choice of how the invariant measures $m_{S_1}$ and $m_{S_2}$ are chosen with respect to each other. Intuitively, this freedom is reflective of the \emph{relative scale} with which the variables are considered: Since ultimately, probability densities are defined by what finite probability masses they integrate to in small neighbourhoods, the relative scale between two variables defines how to compare the neighborhoods for two variables. If the two marginal distributions of the sources are the same, a canonical choice is to take the same partitioning for both variables. On the other hand, if the variables come from very different distributions, the invariant measures $m_{S_i}$ need to be selected to account for that.

On a practical note, the division by the invariant measures $m_{S_i}$ can be absorbed into the definition of the probability densities $f_{S_i}$ by a suitable variable pretransformation. How the variables should be preprocessed, if at all, depends on the nature of the variables in question. Firstly, if the variables describe equivalent source processes, no pretransformation is necessary, and linear pretransformations such as standardization do not make a difference (see \cref{apx:preprocessing} for a more detailed comparison). On the other hand, if two variables measure the same physical quantity but with different scales, e.g., body height measured in centimeters versus in feet, a reasonable preprocessing step would be the rescaling of one of the variables or---if the variables come from a similar distribution---a standardization of both. Alternatively, if the measurement precision, i.e., the relative pointwise deviation, can be assumed to be proportional to the probability density at that point, the original probability densities $f$ may be replaced with the copula densities $c$~\cite{copula_theory}, which furthermore happens to make the measure invariant under invertible mappings of individual variables before this step.

Finally, these considerations lead to the following definition for the probability of a continuous disjunction for two random variables $S_1$ and $S_2$, which have possibly been preprocessed beforehand.
\begin{definition} \label{definition:density_or_two_surces}

Let $S_1, S_2$ be continuous random variables with associated densities $f_{S_1}, f_{S_2}$. Then we define the quasi-density of the disjunctive logic statement $S_1 = s_1 \lor S_2 = s_2$ to be
\begin{align*}
    f_{S_1 \lor S_2}(s_1 \lor  s_2) := f_{S_1}(s_1) + f_{S_2}(s_2) \, .
    \label{eq:definition_or}
\end{align*}

\end{definition}
Note that in the above definition, it is assumed that the precision of one variable does not depend on the value of the other variable and that the relative precision between the two variables is constant. Furthermore, while the local density of a disjunction as defined above does not fulfill all properties of a density function in the mathematical sense (particularly, it does not integrate to one), it nevertheless serves as a local approximate proportionality of probability in relation to magnitude of area in state space.

Given this definition for a continuous disjunction, we can write down an analytical formulation for the continuous local shared-exclusion redundancy.

\begin{definition} \label{definition:local_redundancy}

Let $S_1, S_2$ be continuous random variables with associated densities $f_{S_1}, f_{S_2}$.

Then we define the continuous local shared-exclusion redundancy for two sources as
\begin{gather*}
   i^\mathrm{sx}_\cap(t, \{s_1\}\{s_2\}) 
   := \log_2 \frac{f_{T| S_1 \lor S_2}(t| s_1 \lor s_2)}{f_T(t)}\\
   = \log_2 \frac{f_{T, S_1}(t, s_1) + f_{T, S_2}(t, s_2)}{f_T(t)(f_{S_1}(s_1) + f_{S_2}(s_2))}.
   \label{eq:bivariate_local_sxpid}
\end{gather*}
\end{definition}

This measure for continuous redundancy shares many of the favourable properties of its discrete counterpart. In particular, it is localizable (i.e., the global measure can be expressed as an expected value of local values) and can be generalized to arbitrary numbers of source variables (see~\cref{sec:nd_pid}). Furthermore, the continuous $i^\mathrm{sx}_\cap$ measure is differentiable with respect to the underlying probability density function (see~\cref{apx:differentiability}).

Despite being initially inspired by the ansatz from~\cite{schick-poland}, \cref{definition_1} deviates from the measure-theoretic derivations. The main difference is the absence of the auxiliary random variable designed in~\cite{schick-poland} to carry the logical structure of the setting. This variable, however, stems from a perspective of abstractness that disregards all differences in the measurement scales between the individual density functions and relies on an auxiliary random variable carrying all the irregularity in this existence theorem. In actual applications this perspective needs to be concretized and made evaluable, which introduces the non-uniqueness of the definition due to the measurement scales and preprocessing. Thus \cref{definition:density_or_two_surces} is a single choice of concretization that can be shown to vary smoothly with the overall density as shown in \cref{appendix:derivation_density_logical_statements}.

Comparing this definition for continuous local shared-exclusion redundancy to the discrete local shared-exclusion redundancy of \cref{eq:discrete_bivariate_local_sxpid}, which can be rewritten as
\begin{gather*}
   i^\mathrm{sx, discrete}_\cap(t, \{s_1\}\{s_2\})\\
   = \log_2 \frac{p_{T, S_1}(t, s_1) + p_{T, S_2}(t, s_2) - p_{T, S_1, S_2}(t, s_1, s_2)}{p_T(t)(p_{S_1}(s_1) + p_{S_2}(s_2) - p_{S_1, S_2}(s_1, s_2))},
\end{gather*}
reveals that the continuous definition can be obtained from the discrete one by replacing probability masses $p$ with probability densities $f$ and ignoring the exclusion terms in the disjunctions. Because of these similarities in the analytic form, the derivation from shared exclusions of \citet{makkeh} carries over to the continuous domain, which is why the continuous redundancy measure inherits the same operational interpretation of a memoryless agent trying to guess the target from only the redundant local information on an ensemble of realizations.

As in the case of the mutual information, the average redundancy is defined as the expectation value of the local values, i.e.,
\begin{gather*}
    I^\mathrm{sx}_\cap(T, \{S_1\}\{S_2\})\\
    = \int \mathrm{d}t \, \mathrm{d}s_1 \, \mathrm{d}s_1 \, f_{T,S_1,S_2}(t, s_1, s_2) \, i^\mathrm{sx}_\cap(t, \{s_1\}\{s_2\}).
\end{gather*}

Finally, given this quantification of the redundant information in conjunction with the three classical mutual information terms between $T$ and the sources, the four information atoms can be computed by solving \cref{eq:consistency}, which in the general case is known as a Moebius inversion~\cite{williams2010} (see \cref{sec:nd_pid_lattice}).

\subsubsection{Results for toy examples}
\begin{table*}
\caption{\textbf{Numerical evaluations for four example distributions show qualitative similarities between discrete and the continuous $I^\mathrm{sx}_\cap$ on the three selected logic gates, while qualitative differences can be observed for the sum example.} In the continuous case, the sources are distributed normally, $S_i \sim \mathrm{Normal}(\mu=0, \sigma=1)$, except for $S_2$ in the redundant gate which is set equal to $S_1$, and the target is constructed according to the pictograms with added Gaussian noise ($\sigma = 0.01$) to make the mutual information $I(T:S_1, S_2)$ finite. In the discrete setting, the variables are drawn from uniform binary distributions and no noise is added to the target. 
\label{tab:analytical_examples}}
\begin{ruledtabular}
    \begin{tabular}{rcccc}
        &  \multicolumn{1}{c}{Redundant Gate} & \multicolumn{1}{c}{Copy Gate} & \multicolumn{1}{c}{Unique Gate} & \multicolumn{1}{c}{Sum}\\
        \hline
        \\[-0.2cm]
        \includegraphics[width=0.12\textwidth]{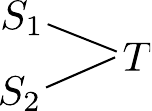}
        & \includegraphics[width=0.12\textwidth]{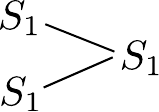}
        & \includegraphics[width=0.15\textwidth]{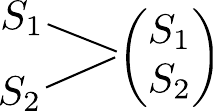}
        & \includegraphics[width=0.12\textwidth]{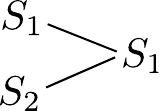}
         & \includegraphics[width=0.17\textwidth]{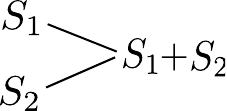}\\
        \\[-0.2cm]
        \hline
        Continuous Gates\\
        \hline
        $\Pi_\mathrm{syn}$ & $0.000 \, \mathrm{bits}$ & $6.644 \, \mathrm{bits}$ & $5.525  \, \mathrm{bits}$ & $6.497  \, \mathrm{bits}$\\
        $\Pi_\mathrm{unq,1}$ & $0.000    \, \mathrm{bits}$ & $0.000 \, \mathrm{bits}$ & $1.119 \, \mathrm{bits}$ & $0.147  \, \mathrm{bits}$\\
        $\Pi_\mathrm{unq,2}$ & $0.000    \, \mathrm{bits}$ & $0.000 \, \mathrm{bits}$ & $-5.525 \, \mathrm{bits}$ & $0.147  \, \mathrm{bits}$\\
        $\Pi_\mathrm{red}$ & $6.644 \, \mathrm{bits}$ & $6.644 \, \mathrm{bits}$ & $5.525   \, \mathrm{bits}$ & $0.353  \, \mathrm{bits}$\\
        $I(T:S_1, S_2)$ & $6.644 \, \mathrm{bits}$ &  $13.288 \, \mathrm{bits}$ & $6.644 \, \mathrm{bits}$ & $7.144  \, \mathrm{bits}$\\
        \hline
        Discrete Gates\\
        \hline
        $\Pi_\mathrm{syn}$ & $0.000 \, \mathrm{bits}$ & $0.415 \, \mathrm{bits}$ & $0.415  \, \mathrm{bits}$ & $0.415  \, \mathrm{bits}$\\
        $\Pi_\mathrm{unq,1}$ & $0.000 \, \mathrm{bits}$ & $0.585 \, \mathrm{bits}$ & $0.585 \, \mathrm{bits}$ & $0.585  \, \mathrm{bits}$\\
        $\Pi_\mathrm{unq,2}$ & $0.000 \, \mathrm{bits}$ & $0.585 \, \mathrm{bits}$ & $-0.415     \, \mathrm{bits}$ & $0.585  \, \mathrm{bits}$\\
        $\Pi_\mathrm{red}$ & $1.000 \, \mathrm{bits}$ & $0.415 \, \mathrm{bits}$ & $0.415   \, \mathrm{bits}$ & $-0.085  \, \mathrm{bits}$\\
        $I(T:S_1, S_2)$ & $1.000 \, \mathrm{bits}$ &  $2.000 \, \mathrm{bits}$ & $1.000 \, \mathrm{bits}$ & $1.500  \, \mathrm{bits}$\\
    \end{tabular}
\end{ruledtabular}
\end{table*}

To gain an intuition for the continuous $I^\mathrm{sx}_\cap$ measure, we apply the analytical formulation established in the previous section to simple and intuitive toy examples. Such simple example cases can be found in logic gates, which have been employed as examples for discrete PID measures throughout the literature  (e.g., \cite{broja,griffith}). We here use continuous versions of the well-established ``redundant'', ``unique'' and ``copy'' gate as well as a sum between source variables and discuss the results compared to their discrete counterparts. In the discrete domain, the two source variables are binary with two equiprobable realizations, whereas in the continuous domain, we  draw the variables from a standard normal distribution. To ensure a finite mutual information, which is a prerequisite for a well-defined continuous PID, Gaussian noise with standard deviation $\sigma = 0.01$ is added to the target variable $T$ in the continuous case. While the choice of $\sigma$ significantly affects the numerical results, the qualitative analysis of the results remains unaffected by the precise value chosen (see \cref{apx:examples_distributions}).

While some logic gates (like the \enquote{and} gate) have no canonical continuous counterparts, we selected three gates whose basic concept can be straightforwardly transferred to the continuous domain: The redundant gate, in which the two sources are the same and equal to the target, the copy gate, in which the target is a two-dimensional variable whose components consist of the two source variables, and the unique gate, in which the target is a copy of the first source variable while the second source variable is drawn independently from the target. In addition to the three logic gates, we added the ``sum'' example, in which the source realizations are interpreted as numerical values and the target variable is the sum of the two sources.

While the computation of the discrete PID atoms is a straightforward exercise, the computation of the continuous atoms is more difficult due to the integral expressions. Note, first, that all continuous examples can be expressed as multivariate Gaussian distributions, for which an analytical closed form for the joint and marginal mutual information terms exists (see \cref{apx:examples_distributions}). The integrals involved in computing the redundancy, however, cannot be easily evaluated analytically. Therefore, we resorted to simple Monte Carlo integration techniques to numerically estimate the involved integrals. To compute the expectation value, we drew randomly from the given joint probability distributions, and computed the local $i^\mathrm{sx}_\cap$ redundancy for $2 \cdot 10^8$ samples. The PID atoms can then be computed from the analytical mutual information terms (see for instance \cref{apx:examples_distributions}) and the numerically computed redundancy via the consistency equations (see \cref{eq:consistency}).

\cref{tab:analytical_examples} shows the results for all four examples. For the continuous redundant gate, $I^\mathrm{sx}_\cap$ attributes all mutual information between the sources and the target to redundancy. This result, which can equivalently be observed in its discrete counterpart, is expected from an information measure constructed for quantifying redundant information, and stems from the fact that the probability of a disjunction between two equal events, e.g. $S_1=s_1 \lor S_1=s_1$ is trivially equal to just the probability the individual event.

In the copy gate, the two marginal mutual information terms each contribute half of the total mutual information, i.e., $I(T:\bm S) = 2I(T:S_1) = 2I(T:S_2)$. Since the two source variables are independent, learning about source event $S_1=s_1$ provides no information about the state of the second source, and vice versa, resulting in no \emph{source redundancy}~\cite{harder2013bivariate}. Nevertheless, the redundancy measured by $I^\mathrm{sx}_\cap$ is non-zero, since the two variables inform redundantly about the target event because of the mechanism with which the target is constructed from the sources (i.e., \emph{mechanistic redundancy}~\cite{harder2013bivariate}): When an agent is informed that either $S_1=s_1$ or $S_2=s_2$ occurred, the agent can redundantly exclude all target events $T = t = (s_1', s_2')$ for which both $s_1' \neq s_1$ (which is excluded by observing $S_1=s_1$) and $s_2' \neq s_2$ (which is excluded by observing $S_2=s_2$). In contrast, for the marginal mutual information $I(T:S_1)$, the hypothetical agent is informed which of the two source events actually took place, e.g., $S_1=s_1$, which enables the agent to additionally rule out all events for which $s_1' \neq s_1$ but $s_2'=s_2$. In the discrete case, these additional events for which $s_1$ does not match the first component of the target event $t$ but $s_2$ matches the second, carry a significant probability mass, which makes the redundancy smaller than the mutual information, which thus produces non-zero unique information. For the continuous copy gate, however, the additionally excluded regions carry only negligible probability, which results in negligible unique information. 

The continuous unique gate produces similar PID results as its discrete counterpart. Despite $S_2$ being independent from $S_1$ and $T$, a significant amount of the joint mutual information is assigned to redundancy, which is attributable to incidental overlaps between exclusions in probability space. Note that since only the first variable provides the full information about the target, i.e., $I(T:\bm S) = I(T:S_1)$, while $I(T:S_2)=0$, the redundancy $\Pi_\mathrm{red}$ and the unique information $\Pi_\mathrm{unq,1}$ add up to the total mutual information $I(T:\bm S)$, resulting in a positive unique contribution of $S_1$. Conversely, the redundancy $\Pi_\mathrm{red}$ and the unique information $\Pi_\mathrm{unq,2}$ need to cancel exactly, resulting in negative unique information for the uncorrelated source $S_2$. Since the unique information of $S_2$ also cancels the redundant information in the consistency equation for the full mutual information, the synergy finally needs to be equal to the redundancy to compensate. While the quantitative results differ, these results are qualitatively similar to the PID of the discrete unique gate.

Finally, for the sum gate, most of the information about the target is carried synergistically between the two source variable, reflecting the fact that knowledge of one variable hardly restricts the possible target realizations that can be taken on. Some information is unique to either of the two sources, since knowledge of just one value already makes some values for the sum less plausible, while the small redundancy in this example can likely again be attributed to mechanistic effects. For this gate, the discrete and continuous examples differ qualitatively: While in the discrete case, the unique atoms are positive and the redundancy is negative, all atoms are positive in the continuous case. These discrepancies are likely due to the fact that the discrete sum gate is not a good analogue of the continuous case, since the values of the realizations the discrete random variable can take on has no natural total ordering, while neighborhoods between values are crucial to the continuous definition.

Depending on the setting that the variables arise in, different preprocessing schemes might need to be applied to the variables before computing the PID atoms. Using different preprocessing schemes, one finds that the values agree qualitatively on most gates, which is mostly due to them being drawn from the same distributions. For a more detailed discussion, refer to \cref{apx:preprocessing}.

\subsection{KSG nearest-neighbour based estimator for mutual information}
\label{sec:estimator}

\begin{figure*}
    \centering
    \includegraphics[width=0.8\textwidth]{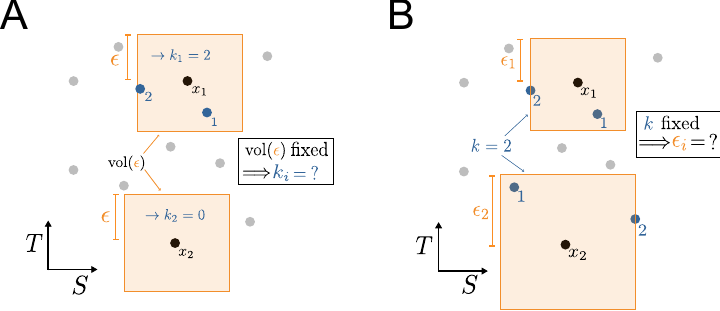}
    \caption{\textbf{Local densities can be estimated from data in two opposite ways.} \textbf{A} Kernel density estimation using a fixed volume $\mathrm{vol}(\epsilon)$ around the reference point $x_i$. \textbf{B} $k$-nearest-neighbor-based estimation using an adaptive search volume, determined by the distance to the $k$-th neighbor.}
    \label{fig:kde_vs_kNN}
\end{figure*}

While we established in \autoref{sec:analytical} how continuous $I^\mathrm{sx}_\cap$ can be computed from a given probability distribution, in most research settings this distribution is unknown. More commonly, one has access to only a finite sample drawn from this distribution and needs to \emph{estimate} the information-theoretic quantities of interest from this sample. For mutual information, several different approaches to this estimation problem have been discussed in literature, for instance, binning or kernel density estimation (see, e.g., \cite{hlavavckova} for a review). However the nearest-neighbour-based approach introduced by Kraskov, Stögbauer and Grassberger~\cite{kraskov} (KSG) has been widely adopted due to its favourable variance-bias characteristics. In the following, we explain how local probabilities can be estimated using nearest-neighbor techniques~\cite{kl} and introduce the key ideas behind the KSG estimator. Then, we show how these concepts can be built upon to derive an estimator for continuous $I^\mathrm{sx}_\cap$.

Many information-theoretic quantities like the mutual information (see \autoref{eq:sum_mi}) can be expressed as expectation values of \emph{local} quantities. For discrete data, a simple (but not optimal~\cite{nemenman2001entropy}) way to approximate these local values is to replace the probability masses by their frequencies in the sample, which is referred to as \textit{plug-in} estimation. In the continuous case, however, all samples are almost always unique, which is why no continuous analogue to the plug-in principle exists. Thus, one has to revert to more elaborate methods to estimate the local densities $f(x_i)$ necessary for computing the local information-theoretic values.

\paragraph{Nearest-neighbor-based estimation of local probability densities}

To estimate the local probability densities at the sample points, we exploit the intuitive idea that sample points are more likely to be closer together in regions of higher probability density. Note that this intuition crucially rests on the assumption that the sample is sufficiently large such that the true probability density varies only little between neighbouring data points. If this prerequisite is met, a natural approach to estimate the local probability density is to take a ball $B_\epsilon(x_i)$ of fixed radius $\epsilon$ around the evaluation point $x_i$, over which the probability density is assumed to be approximately constant, and count the number of neighbours $k_i$ within the ball (see \autoref{fig:kde_vs_kNN}.\textbf{A}). A point estimate for the probability density is then given by the fraction of sample points $(k_i + 1)/N$ within that $\epsilon$-ball divided by its volume $\mathrm{vol}(B_\epsilon(x_i))$ as 
\begin{equation}
    \hat f(x_i) = \frac{k_i + 1}{N \mathrm{vol}(B_\epsilon (x_i))} \, .
    \label{eq:p_estimate_kde}
\end{equation}

Kernel density estimators (KDE) use this principle, but typically weigh the neighbors by their respective distance using a fixed \emph{kernel}, which in general does not have to be the characteristic function of an $\epsilon$-neighborhood.

The disadvantage of approaches like this is that the same ball or kernel is used irrespective of the local density: In areas of low density there may be insufficiently many points within the neighborhood to arrive at a good estimation, while conversely in areas of high density the high numbers $k_i$ would have allowed for a smaller search region, partly relaxing the smoothness assumption on $f$. A way to get around these problems is to invert the procedure: Instead of fixing the search radius $\epsilon$ and counting the number of neighbouring points $k_i$, one can fix a natural number $k$ and compute the distance $\epsilon_i$ from point $x_i$ to its $k$-th nearest neighbor (see \autoref{fig:kde_vs_kNN}.\textbf{B}). Analogously to \autoref{eq:p_estimate_kde}, a point estimate of the local probability density can now be expressed as
\begin{equation}
    \hat f(x_i) = \frac{k + 1}{N \mathrm{vol}(B_{\epsilon_i} (x_i))} \, .
    \label{eq:p_estimate}
\end{equation}

\paragraph{The Kozachenko-Leonenko Estimator for Shannon Differential Entropy}
\citet{kl} used the nearest-neighbour-based approach to build an estimator for the Shannon differential entropy~\cite{coverthomas}
\begin{equation}
    \label{eq:h}
    H(X) =\mathbb{E}[-\log_2p(X)]=-\int \mathrm{d}x\, p(x)\log_2 p(x) \,.
\end{equation}
However, instead of using the point estimate of \autoref{eq:p_estimate}, the authors derive the expectation value for $\log p(x_i)$ for the probability mass of the $\epsilon$-ball 
\begin{equation*}
p(x_i) = \int_{B_{\epsilon_i}(x_i)}\mathrm{d}x' \, f(x') \approx \mathrm{vol}(B_{\epsilon_i}(x_i)) f(x_i)
\end{equation*}
given the number of neighbours $k$ for each sample point, which results in the expression
\begin{equation*}
    \widehat{\log p_X(x_i)} = \psi(k) - \psi(N) \, ,
\end{equation*}
where $\Psi$ represents the Digamma function.
Plugging this result into \autoref{eq:h}, the authors obtain an estimator for the Shannon differential entropy as (for a more formal derivation, refer to \autoref{apx:kl})
\begin{equation}
\begin{gathered}
\hat{H}(X)
= -\frac{1}{N} \sum\limits_{i=1}^N \widehat{\log\left[p_{X}(x_i)\right]}\\
= - \psi(k) + \psi(N) + \frac{1}{N} \sum\limits_{i=1}^N \log\left[\mathrm{vol}\left(B_{\epsilon_i}\right)\right] \, .
\end{gathered}
\label{eq:kl}
\end{equation}

\paragraph{The Kraskov-St{\"o}gbauer-Grassberger (KSG) Estimator for Mutual Information}
\begin{figure*}[ht]
    \centering
    \includegraphics[width=\textwidth]{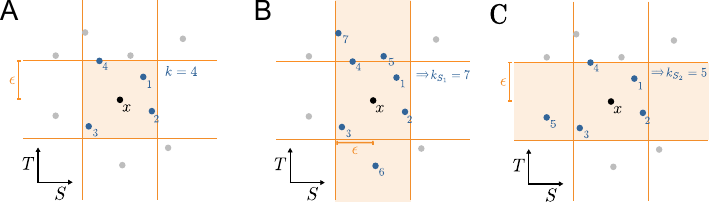}
    \caption{\textbf{The KSG estimator for mutual information~\cite{kraskov} works by considering $k$th nearest neighbors in the joint and marginal spaces in three steps.} \textbf{A} Determining search radius in joint space \textbf{B} counting neighbors in marginal space of $S$ and \textbf{C} counting neighbors in marginal space of $T$.}
    \label{fig:kraskov}
\end{figure*}
The most straightforward way to construct an estimator for continuous mutual information from the results of \citet{kl} is by using the identity $I(T:S) = H(T) + H(S) - H(T, S)$. However, in doing so, the biases of the individual entropy estimates likely do not cancel, which led \citet{kraskov} to take a different approach: First, like for the Kosachenko-Leonenko estimator, the distance $\epsilon_i$ of point $x_i = (t_i, s_i)$ to its $k$th nearest neighbour is determined in the joint space (see \autoref{fig:kraskov}.\textbf{A}). Subsequently, the number of neighbours $n_T$ which have a distance of less than $\epsilon_i$ in only the marginal space of the first variable $T$ is determined (see \autoref{fig:kraskov}.\textbf{B}), and the procedure is repeated analogously to determine the number of neighbors $n_S$ in radius $\epsilon_i$ in the marginal space of the second random variable $S$ (see \autoref{fig:kraskov}.\textbf{C}). The advantage of obtaining the marginal probability densities using the distance determined in the joint space is that---when using the maximum norm---the volume terms which appear in \autoref{eq:kl} cancel exactly, leading to the succinct expression
\begin{equation}
    \hat{I}(T:S) = \Psi(k) + \Psi(N) - \left<\Psi(n_T(i)) + \Psi(n_S(i))\right>_i \, ,
\end{equation}
for the estimated mutual information. Here, angled brackets $\left<\cdot\right>_i$ denote the average over the sample points.

Clearly, the estimated density as well as the information quantity assessed via the density estimation will depend on $k$. However, Kraskov et al. \cite{kraskov} have empirically studied the influence of $k$, or more exactly a fixed $\frac{k}{N}$ on the estimation result and concluded that in the limit of infinitely many sample points the exact choice of $k$ does not matter as all choices will converge. 

\subsection{Generalizing the KSG estimator for $I^{\mathrm{sx}}$}
\label{subsection:generalizing_kraskov}

\begin{figure*}
    \centering
    \includegraphics{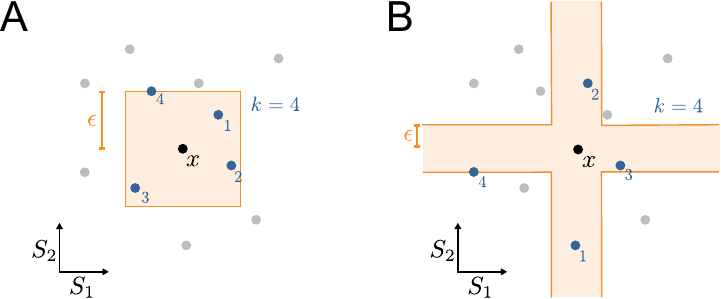}
    \caption{\textbf{To estimate local redundant information for a two-variable example, the search regions of the original mutual information KSG estimator~\cite{kraskov} need to be adapted.} \textbf{A} For mutual information estimation, the radius $\epsilon$ at a point $x$ is assessed in the joint space by searching for the minimal distance such that \emph{all} of the variables' marginal distances to $x$ is smaller or equal to $\epsilon$. \textbf{B} For redundancy estimation, on the other hand, $\epsilon$ is inferred by searching for the minimal distance such that \emph{at least one} of the variables' marginal distances to $x$ is smaller or equal to $\epsilon$.}
    \label{fig:ksg_vs_isx}
\end{figure*}

\paragraph{Adapting local probability density estimation for redundant information}
To understand how to estimate the continuous redundancy $I^\mathrm{sx}_\cap$ using an approach similar to \citet{kraskov}'s, first consider the case of estimating the mutual information $I(T:\bm S)$ between a single target random variable $T$ and two joint source variables $\bm S = (S_1, S_2)$ using the KSG estimator. In the first step, the radius $\epsilon_i$ is determined as the smallest radius in the joint space for which the $k$-th neighbor of the point $x_i = (t_i, s_{1, i}, s_{2, i})$ lies within the Ball $B_{\epsilon_i}(x_i)$. Using the maximum norm, this ball in the joint space is equal to the product of the $\epsilon$-balls in the marginal spaces, i.e.,  $B_{\epsilon_i}(t_i) \cross B_{\epsilon_i}(s_{1, i}) \cross  B_{\epsilon_i}(s_{2, i})$ (square region in \autoref{fig:ksg_vs_isx}.\textbf{A}). In the second step, neighbours are counted within the marginal space of $\bm S$, where only those points are counted which lie within less than the distance $\epsilon_i$ to $x_i$ in both dimensions $S_1$ \emph{and} $S_2$ simultaneously.

As laid out in detail in \autoref{sec:analytical}, the difference in the mathematical formulation between mutual information between a single target $T$ and two (or more) source variables $\bm S = (S_1, S_2)$ and $I^\mathrm{sx}_\cap$ redundancy is that while the former is expressed with conjunctions, which amount to \emph{intersections} of source events in sample space, the latter deals with disjunctions, which result in \emph{unions}. The KSG estimator can be readily adapted to this concept by changing the search regions accordingly: To compute the redundancy in the example case of two source variables introduced above, the radius $\epsilon_i$ now has to be determined not as the smallest radius laying on the intersection of $B_{\epsilon_i}(t_i)$, $B_{\epsilon_i}(s_{1, i})$ and $B_{\epsilon_i}(t_{2, i})$ (``square'' region in \autoref{fig:ksg_vs_isx}.A) but on the union of the source variable neighbourhoods intersected with the target ball $B_{\epsilon_i}(t_i) \cross B_{\epsilon_i}(s_{1, i}) \cross S_2 \cup B_{\epsilon_i}(t_i) \cross S_1 \cross B_{\epsilon_i}(s_{2, i})$ (``cross'' region in \autoref{fig:ksg_vs_isx}.\textbf{B}). Similarly, $n_S$ is determined by counting all points which lie within the union of marginal balls with radius $\epsilon$.

Note that, analogously to Kraskov's~\cite{kraskov} estimator for mutual information, this estimation procedure assumes the maximum norm when combining the unions between the source with the target variables in order for the volume terms of the KL estimator to cancel, which simplifies the analytical formulation considerably. Although the choice of any particular distance norm will matter for small sample sizes (e.g. Gao et al.\cite{gao2018demystifyingkNN} suggest using the Euclidean norm), the assumed continuity of the density function and all its marginals ensures that regardless of the choice of distance norm the estimation of the local densities \emph{asymptotically} converges to the same values. Intuitively, the average probability over decreasing neighbourhoods of a point converges to the density at that point in the limit of infinitely many samples, regardless of the exact choice of the shape of the neighbourhood, the shape being determined by the distance norm. Technically, however, these neighbourhoods must not be chosen entirely arbitrarily--the shrinking procedure obey the \emph{nicely shrinking} \cite{rudin86real} property, guaranteeing that the shrinking neighbourhoods indeed cover neighbourhoods in each dimension of the point. For the nicely shrinking property to hold, one can choose any distance norm as all imply the nicely shrinking property. This asymptotic convergence behavior of the estimated densities then directly implies the asymptotic convergence of the estimate of the redundant information to its analytical value.

\paragraph{Nearest-neighbor-based estimator for continuous redundant information}
Adapting the steps by \citet{kraskov} using the described estimation procedure for the probability density of a logical disjunction we find
\begin{equation}
\begin{gathered}
    \hat{I}^{sx}_\cap(T: X, Y)
    = \langle\widehat{\log[f^i_{T, X \lor Y}(t, x, y)]}\rangle_i\\
    - \langle\widehat{\log[f^i_{X \lor Y}(x, y)]}\rangle_i - \langle\widehat{\log[f^i_T(t)]}\rangle_i \\
    = \psi(k) + \psi(N) - \langle \psi(n_{X \lor Y}(i))\rangle_i - \langle \psi(n_T(i))\rangle_i \, .
\end{gathered}
\label{eq:exponential_expansion_estimation}
\end{equation}

More precisely, we have retaken the approach of choosing the $\epsilon$ to be determined by a nearest neighbor search in the joint space. We then use the same $\epsilon$ in the marginal spaces themselves, such that the volume terms cancel exactly. The search with the predetermined $\epsilon$ will then cause an adapted number of nearest neighbors in the marginal spaces, denoted by $n_{X \lor Y}(i)$ and $n_T(i)$, respectively.
The exact steps leading to \cref{eq:exponential_expansion_estimation} are outlined in \cref{appendix_b_estimation_steps}. Furthermore, details of the implementation and code availability are described in \cref{apx:implementation}.

Note, however, one subtle difference in the meaning of the parameter $\epsilon$ between the mutual information and redundancy case: While scaling $\epsilon$ by a constant factor for one but not the other source variable changes the estimated mutual information for any finite sample, it does not change what the mutual information estimator converges to in the limit of infinite samples. This is different in the case of redundancy: Because of the addition of marginal densities in \cref{definition:density_or_two_surces}, scaling one but not the other search radius $\epsilon$ makes a difference also in the limit of infinite samples. Using the same $\epsilon$ for both source variables here thus reflects the choice of treating the source variables on ``equal footing''. This train of thought aligns with the dependence of the result on the chosen distance function, as in many cases a transformation of one of the variables can be reflected with a change in the utilized norm.

For a fixed small $k$, the specific value of $k$ does not influence the estimation results if sufficient samples $N$ are provided. Analogously to the KSG estimator for mutual information, the estimation results only depend on the fraction $k/N$, as a fixed fraction $k/N$ hinders the estimator in resolving high-frequency fluctuations in the probability density. Fluctuations on all length scales are resolved in the limit $k/N \to 0$, for which the estimator converges to the analytical solutions (see~\cref{fig:comparison_k} and \cref{fig:comparison_kN} in \cref{apx:k}). Following the suggestions from \citet{kraskov}, we typically recommend $k$ to be set to $2$, $3$ or $4$, where estimates with higher $k$ have smaller variance due to the stochasticity of the distances neighbours averaging out, but also higher bias due to their inability to resolve small-scale features.

\subsection{Convergence of the estimator for simple toy examples}

\begin{figure}[h!]
    \centering
    \includegraphics[width=\columnwidth]{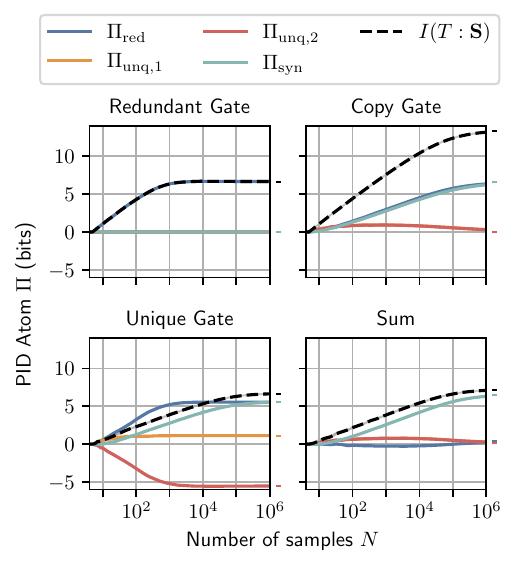}
    \caption{\textbf{With increasing sample size $N$, the estimated PID atoms converge to the numerically evaluated analytical results.} The subfigures show the four PID atoms $\Pi$ and the mutual information $I(T:\bm{S})$ for the continuous redundant, unique and copy gates as well as the sum example estimated using the nearest-neighbors approach outlined above with $k=4$. The colored bars to the right of the graphs correspond to the respective analytical value for each atom and gate from \cref{tab:analytical_examples}.}
    \label{fig:convergence}
\end{figure}

In order to assess the efficacy of our estimation procedure, we applied it to the toy examples outlined in \cref{sec:analytical}. To this end, different numbers of samples have been drawn from the known probability distributions and have been plugged into the estimator to check for convergence.

\cref{fig:convergence} shows how the PID atoms converge with more samples for the four example cases. The copy gate and sum example converge slower than the other values. This behaviour is inherited from the KSG estimator, which also converges more slowly for higher-dimensional distributions and distributions with higher mutual information. All in all, we conclude that the estimator succeeds in reproducing the numerically evaluated analytical results from \cref{sec:analytical} given sufficiently many samples.

\section{Definition of multivariate continuous $I^\mathrm{sx}_\cap$}
\label{sec:nd_pid}
\subsection{The multivariate PID Problem}
\label{sec:nd_pid_lattice}

In the previous sections we introduced the continuous $I^\mathrm{sx}_\cap$ measure for two source variables. However, as was already mentioned, the PID framework as envisioned by \citet{williams2010} is more general and in principle applicable to an arbitrary number of sources. In the following, we will explain the PID problem for the general, multivariate case and generalize the continuous $I^\mathrm{sx}_\cap$ measure accordingly. This section may be skipped if the reader has a particular application in mind for which two source variables suffice and this section is found too technical.

\begin{figure*}[t]
    \centering
    \includegraphics[width=0.7\textwidth]{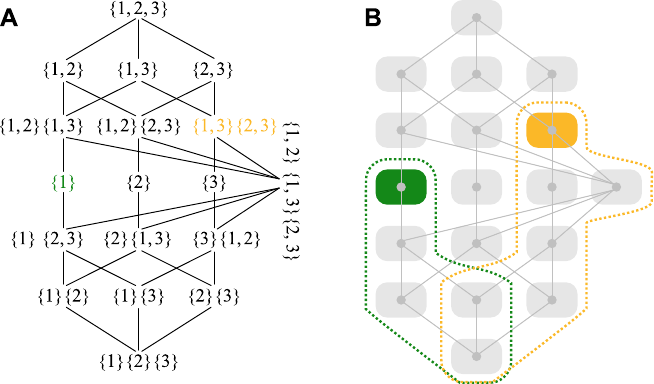}
    \caption{
        \textbf{The PID lattice structure for three source variables reveals how generalized redundancies can be dissected into PID atoms~\cite{williams2010,gutknecht2021}.}
        \textbf{A} Every node in the redundancy lattice represents a generalized redundancy $I_\cap$ for a given antichain, while edges indicate parthood of generalized redundant information in information terms higher up the lattice. For instance, the mutual information $I(T:S_1)$, which corresponds to the self-redundancy with the antichain $\{1\}$, contains all redundancies in its downset, i.e., those with antichains $\{1\}\{2, 3\}$, $\{1\}\{2\}$, $\{1\}\{3\}$ and $\{1\}\{2\}\{3\}$, which, however, themselves overlap and contain parts of each other.
        \textbf{B} The PID atoms (shaded boxes), which are referenced by the same antichains, only comprise the information that the redundancy of the same antichain contains \emph{in excess of what the all the atoms below contain together}. Atoms thus reflect the incremental increases of the redundancies which conversely can be used to construct the generalized redundancies. For instance, the redundancy with antichain $\{1\}$ is the sum of the four atoms in its downset plus the atom $\{1\}$ itself (green shaded box and green dashed outline, respectively). Analogously, the figure shows how the redundancy with the antichain $\{1, 3\}\{2, 3\}$ (orange) is constructed from the atoms in its downset.
    }  
    \label{fig:lattice}
\end{figure*}

Before we generalize to an arbitrary number of source variables, note that already in the bivariate case, each atom can be uniquely identified by which mutual information terms it contributes to via the consistency equations (see \cref{eq:consistency}). For instance, the unique atom $\Pi_\mathrm{unq,1}$ is the only atom that is part of both $I(T:S_1)$ and $I(T:\bm S)$ but not $I(T:\bm S_2)$ while the synergistic information atom $\Pi_\mathrm{syn}$ is precisely the one atom whose information is contained in neither $I(T:S_1)$ nor in $I(T:S_2)$, but only in $I(T: S_1, S_2)$. In the following, we show how this notion of parthood, also called \emph{mereology}, can serve as a general framework for identifying and ordering PID atoms for arbitrarily many source variables following the work by~\citet{gutknecht2021}.

Mathematically, the mereological relations between an atom and the marginal mutual information terms can be captured by a \emph{parthood distribution} $\Phi: \mathcal{P}(\bm S) \rightarrow \{0, 1\}$: A boolean function defined on the power set of source variables $\mathcal{P}(\bm S)$ which is equal to $1$ for exactly those sets of sources $\bm S_{\bm a} = \{S_i \mid i \in \bm a \}$ for which the atom $\Pi_\Phi$ is part of the marginal mutual information $I(T:\bm S_{\bm a})$~\cite{gutknecht2021}. Note that not all boolean functions constitute valid parthood distributions: Since all information about $T$ that is contained in the mutual information $I(T:\bm S_{\bm a})$ with a set of sources $\bm S_{\bm a}$ must naturally also be contained in the mutual information $I(T:\bm S_{\bm b})$ with any superset $\bm S_{\bm b} \supset \bm S_{\bm a}$, $\Phi$ has to be \emph{monotone}, in the sense that $\bm S_{\bm b} \supset \bm S_{\bm a} \Rightarrow (\Phi(\bm S_{\bm a}) = 1 \to \Phi(\bm S_{\bm b}) = 1)$. For instance, any information provided by source $S_1$ alone must naturally also be present when taking sources $S_1$ and $S_2$ together. Furthermore, we must always require that $\Phi(\emptyset) = 0$ and $\Phi(\bm S) = 1$, indicating that no atom is part of the mutual information with the empty set and, conversely, all information atoms are part of the full joint mutual information $I(T:\bm S)$.

The marginal mutual information terms can be constructed additively from these PID atoms according to the generalized consistency equations
\begin{equation}
    I(T:\bm S_{\bm a}) = \sum_{\Phi : \Phi(\bm S_{\bm a}) = 1}\Pi_\Phi \, .
    \label{eq:marginal_mi}
\end{equation}

The number of such parthood distributions, and thus of the PID atoms, grows superexponentially like the Dedekind numbers \cite{gutknecht2021}. On the other hand, there are only $2^n-1$ classical mutual information quantities providing constraints through \cref{eq:marginal_mi}, leaving these consistency equations underdetermined and making it impossible to determine the sizes of the PID atoms. This underdeterminedness is typically resolved by introducing a measure for generalized redundancy $I_\cap(T:\bm S_\Phi)$ for each parthood distribution $\Phi$. This redundancy generalizes the mutual information and fulfills

\begin{equation}
    I_\cap(T:\bm S_\Phi) = \sum_{\Psi: \Psi \prec \Phi} \Pi_\Psi \, ,
    \label{eq:redundancy_from_atoms}
\end{equation}
where the sum is over all atoms $\Pi_\Psi$ with a parthood distribution $\Psi$ which are part of the generalized redundancy $I_\cap(T:\bm S_\Phi)$.

This condition is similar to the monotonicity requirement of the mutual information mentioned before and is mathematically captured in the ordering relation $\prec$ given by
\begin{equation*}
\Psi \prec \Phi \Leftrightarrow \Phi(\bm S_{\bm a}) = 1 \to \Psi(\bm S_{\bm a}) = 1 \text{ for all } \bm S_{\bm a} \subseteq \bm S \, .
\end{equation*}
This partial ordering bestows a structure onto the set of PID atoms which is known as the ``redundancy lattice''\cite{williams2010, gutknecht2021} (see \cref{fig:lattice}).

Comparing \cref{eq:marginal_mi,eq:redundancy_from_atoms}, one finds that the mutual information can be interpreted as a \enquote{self-redundancy}---a special case of a redundancy which describes how the mutual information provided by a set of sources about the target is trivially redundant with itself. Finally, to quantify the atoms $\Pi_\Phi$ given the redundancies $I_\cap(T:\bm S_\Phi)$, \cref{eq:redundancy_from_atoms} can be inverted, which, due to the partial ordering, is known as a \emph{Moebius inversion}~\cite{williams2010}.

Equivalently, each parthood distribution can be uniquely identified by the set of sets of source variables for which the parthood distribution is equal to $1$, i.e., $\Phi^{-1}[\{1\}]$. Because of the monotonicity requirement mentioned before, we can make this representation even more concise by removing all sets which are supersets of others, which gives the so-called \emph{antichains}
\begin{equation}
    \alpha = \left\{\bm a \in  \Phi^{-1}[\{1\}] \mid \not\exists\, \bm b \in  \Phi^{-1}[\{1\}] \text{ s.t. } \bm b \subseteq \bm a \right\} \, ,
\end{equation}
which is the most prevalent way how PID atoms are referred to throughout literature and also the formalism that PID was originally conceived in~\cite{williams2010}. Throughout this paper, the outer braces of the set of sets are neglected for the sake of brevity (e.g., $\{1\},\{2\}$ instead of $\{\{1\},\{2\}\}$. The term antichain stems from the fact that the sets contained within $\alpha$ are incomparable with respect to the partial order of set inclusion.
These antichains make the intuitive meaning of the redundancy terms explicit, e.g., $I_{\cap, \{1\}\{2, 3\}}$ refers to the information that can either be obtained from source $S_1$ alone, or--redundantly to that--from sources $S_2$ and $S_3$ taken together.

\subsection{Definition of multivariate continuous $I^\mathrm{sx}_\cap$}
\label{sec:multivariate_sxpid}
Since in the multivariate case the generalized redundancy $I_{\cap, \alpha}$ refers, for a particular antichain $\alpha$, to information that can equivalently be obtained from all sets $\bm S_{\bm a}$ for $\bm a \in \alpha$ of source variables taken together, the logical statement in the conditional probability of \cref{eq:bivariate_local_sxpid} needs to be generalized to accommodate for such general disjunctions as
\begin{gather}
    \label{eq:local_sxpid}
    i_{\cap, \alpha}^\mathrm{sx}(t: s_1; s_2) = \log_2\left[\frac{p(t | \bigvee_{\bm a \in \alpha} \bm S_{\bm a} = \bm s_{\bm a})}{p(t)}\right]\\
    = \log_2\left[\frac{p(t | \bigvee_{\bm a \in \alpha} \bigwedge_{i \in \bm a} S_i = s_i)}{p(t)}\right] \, ,
\end{gather}
where $\bm S_{\bm a} = \{S_i \mid i \in \bm a\}$ and likewise $\bm s_{\bm a} = \{s_i \mid i \in \bm a\}$. This generalization is analogous to how the multivariate redundancy is introduced for the discrete case by \citet{makkeh}.

\subsection{Estimation of multivariate continuous $I^\mathrm{sx}_\cap$}\label{sec:estimator_multivariate} To estimate multivariate continuous $I^\mathrm{sx}_\cap$, the search regions used by the nearest neighbor estimator need to be expanded to higher dimensions. For $d>2$ source variables $\bm S = (S_1, \dots, S_d)$ and a specific antichain $\alpha$, the search space for determining the radius $\epsilon_i$ is given by the union of intersections 
\begin{equation}
    \bigcup_{\bm a \in \alpha} B_{\epsilon_i}(t_i) \times \bigtimes_{j \in \bm a} B_{\epsilon_i}(s_{i, j}) \times \bigtimes_{j \notin \bm a} S_j \, .
\end{equation}

For illustration, the search regions for three source variables (trivariate PID) are visualized in \cref{fig:trivariate_search_regions} for all antichains in the redundancy lattice.

\makeatletter\onecolumngrid@push\makeatother
\begin{figure*}
    \begin{minipage}[b]{0.67\textwidth}
        \includegraphics[width=\textwidth]{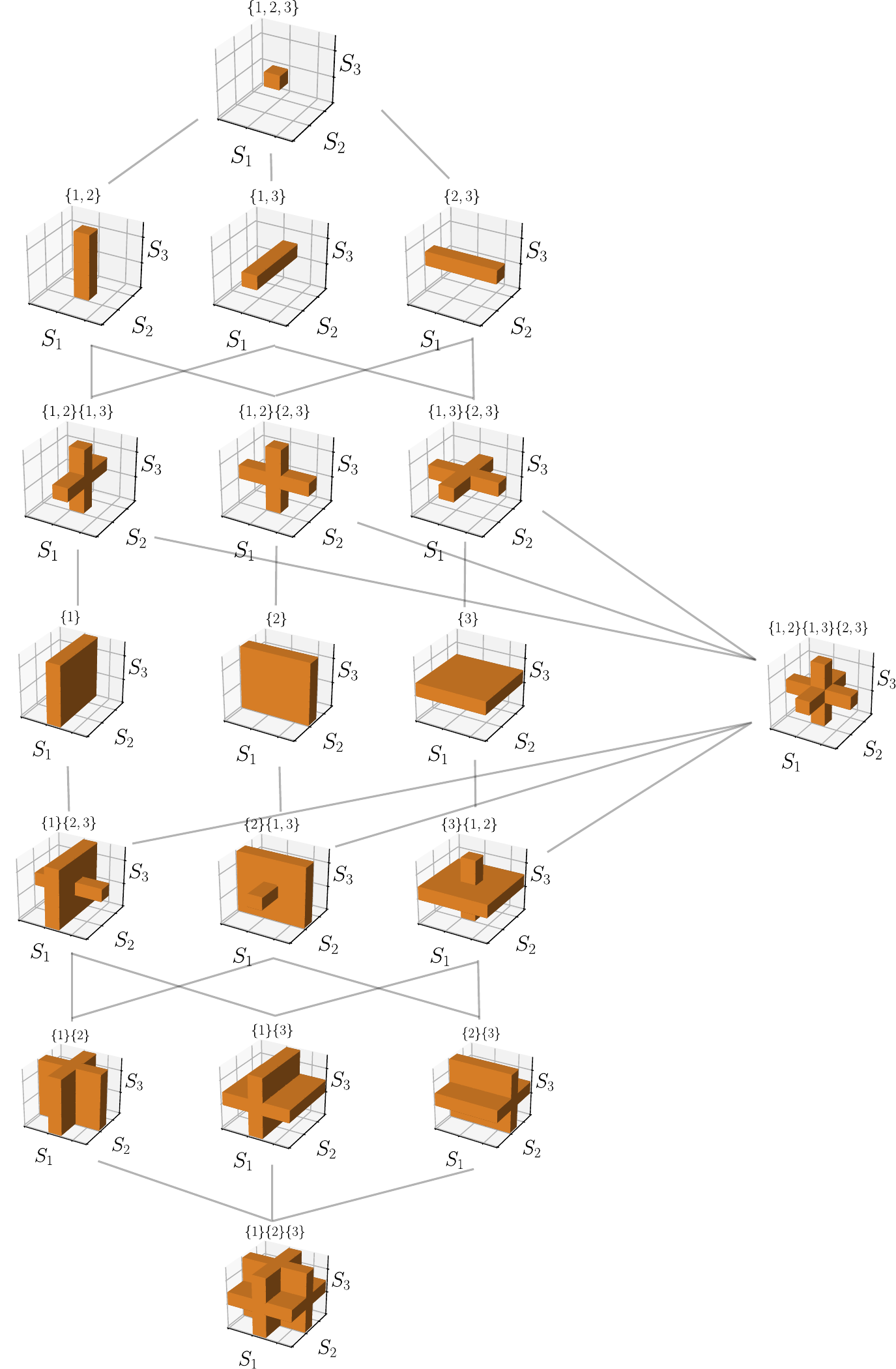}
    \end{minipage}
     \caption{\textbf{The nearest neighbor search regions for three source variables are cumulative when going from the full synergy (top) towards the full redundancy (bottom) in the redundancy lattice.} The orange regions depict where to search for nearest neighbors, and its border signifies the surfaces of equal distance from the logical statement corresponding to each antichain. Here the finite heights, widths and depths $\epsilon$ are chosen equally large. Also note that the search regions for all classical mutual information terms (i.e., antichains $\{1, 2, 3\}$, $\{1, 2\}$, $\{1, 3\}$, $\{2, 3\}$, $\{1\}$, $\{2\}$ and $\{3\}$) correspond to marginalized hypercubes, just as in the regular definition of the KSG estimator~\cite{kraskov}.}
    \label{fig:trivariate_search_regions}
\end{figure*}
\makeatletter\onecolumngrid@pop\makeatother

Using this notion of neighborhood, the estimator for continuous $I^\mathrm{sx}_\cap$ can be generalized to arbitrary antichains $\alpha$ as in \cref{definition_1}, yielding

\begin{equation*}
    \begin{gathered}
        \hat{I}^{sx}_\cap(T: \alpha)\\
        = \langle\widehat{\log[f^i_{\alpha, T}(t, s_k)]}\rangle_i - \langle\widehat{\log[f^i_\alpha(s_k)]}\rangle_i - \langle\widehat{\log[f^i_T(t)]}\rangle_i \nonumber \\
        = \psi(k) + \psi(N) - \langle \psi(n_\alpha(i))\rangle_i - \langle \psi(n_T(i))\rangle_i \, .
    \end{gathered}
\end{equation*}

\section{Application of continuous $I^{\mathrm{sx}}_\cap$ estimator to simulated data from an energy management system\label{sec:application}}

\subsection{Simulation of an energy management system}

\begin{figure*}[tbh]
\centering
\begin{tikzpicture}

\node[inner sep=0pt, label={[align=center]south:Net Energy Transfer\\to External Grid}] (exgrid) at (5.5,0)
    {\includegraphics[width=.1\textwidth]{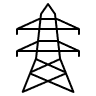}};

\node [fill=black!10,align=center,anchor=west] at ([yshift=0.3cm]exgrid.south east) {$P_\mathrm{Grid}$ ($T$)};

\node[inner sep=0pt, label={[align=center]south:Wind\\Turbine}] (wind) at (1,3)
    {\includegraphics[width=.1\textwidth]{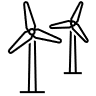}};

\node[inner sep=0pt, label={[align=center]south:Weather\\Data}] (weather) at (5.5,5)
    {\includegraphics[width=.1\textwidth]{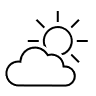}};

\node [fill=black!10,align=left,anchor=west] at (weather.north east) {Ambient temparature $T_A$ ($S_1$)\\Wind speed $v_\mathrm{Wind}$ ($S_2$)\\Time of day $t_\mathrm{Day}$ ($S_3$)};

\node[inner sep=0pt, label={[align=center]south:PV\\System}] (pv) at (4,3)
    {\includegraphics[width=.1\textwidth]{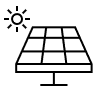}};

\node[inner sep=0pt, label={[align=center]south:Building Power\\Consumption}] (building) at (10,3)
    {\includegraphics[width=.1\textwidth]{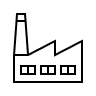}};


\node[inner sep=0pt, label={[align=center]south:Combined Heat\\And Power Plant}] (chp) at (7,3)
    {\includegraphics[width=.1\textwidth]{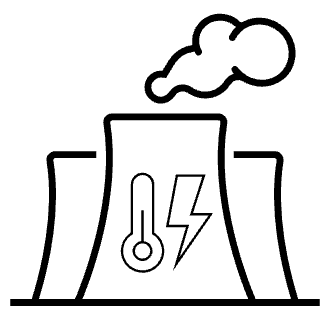}};
    
\node[inner sep=0pt, label={[align=center]south:Battery Storage\\And Controller}] (battery) at (10.6,0)
    {\includegraphics[width=.1\textwidth]{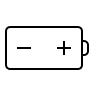}};

\draw [thick, >=latex, ->] (weather.west) -| (wind);  
\draw [thick, >=latex, ->] ([yshift=-0.2cm]weather.west) -| (pv);  
\draw [thick, >=latex, ->] ([yshift=-0.2cm]weather.east) -| (chp);  
\draw [thick, >=latex, ->] (weather.east) -| (building);  

\draw [thick, >=latex, ->] ([yshift=-1.0cm]wind.south) |- (exgrid.west);  
\draw [thick, >=latex, ->] ([yshift=-1.0cm]pv.south) |- ([yshift=0.2cm]exgrid.west);  

\draw [thick, >=latex, ->] ([yshift=-1.0cm]chp.south) |- ([yshift=0.4cm]exgrid.east); 
\draw [thick, >=latex, ->] ([yshift=-1.0cm, xshift=-0.8cm]building.south) |- ([yshift=0.2cm]exgrid.east); 
\draw [thick, >=latex, ->] (battery.west) |- (exgrid); 

\draw [thick, >=latex, ->] ([xshift=0.6cm, yshift=-1.0cm]building.south) -- ([yshift=+0.5cm]battery.center);
\draw [thick, >=latex, ->] (building.west) -- (chp.east);

\end{tikzpicture}
\caption[Energy management system]{
\textbf{Causal structure of the energy management system model.} An arrow from A to B reads as ``A influences B''.
The model simulates the building power and heat demand based on time and weather conditions.
Energy is provided by the grid connection, a PV-system, a combined heat and power plant (CHP), a wind turbine, and a stationary battery.
The battery's charging and discharging behavior is controlled depending on the overall power consumption level.
The CHP produces energy depending on the overall power consumption and the ambient temperature. The gray boxes highlight the choices for source ($S_i$) and target ($T$) variables for the PID analysis.
}
\label{fig:energy_system}
\end{figure*}

To demonstrate the applicability of the proposed definition for continuous $I^\mathrm{sx}_\cap$ and its corresponding estimator, we applied it to simulated data for an energy management system and show how the simulated relationships between system components can be recovered from the estimated PID. In a hybrid simulation approach, we integrated real-world sensor readings from a Honda R\&D facility into simulation components, which were in turn modelled via fundamental physical differential equations~\cite{Unger2012}. We informed our simulation by recordings of the actual facility energy consumption profile and local weather patterns, obtained by smart meter measurement. Specifically, the measured time series consist of all environmental inputs such as temperature, wind speeds, and humidity, as well as the power consumption of subsystems of the Honda R\&D facility, while all other variables were simulated. The simulations have been done in the Modelica simulation language, using the commercial tool SimulationX~\cite{simulationx}.

The simulation comprises multiple different agents representing energy sources, energy storage, consumers and grid interconnections. The energy originates from four sources: An on-premise wind turbine, a photovoltaic (PV) system, a combined heat and power natural gas plant (CHP), as well as energy transfers from the electricity grid. This energy is utilized as building power consumption, while excess power is stored in the battery storage system or fed back into the external power grid (see \cref{fig:energy_system}).

In the simulation, the different agents of the system are given set rules for how they respond to each other or to external changes. Firstly, as the CHP not only provides electrical energy to the facility but also heating, it is set to operate only during times in which the ambient temperature is equal to or lower than \SI{16}{\celsius}. During those times, the plant produces a continuous electric power output depending on the overall energy consumption of the facility (see \cref{alg:chp_controller}). The stationary battery is either charged or discharged depending on the overall power consumption level of the facility (see \cref{alg:battery_controller}). Specifically, the storage system is set to \texttt{Discharge} according to a binary hysteresis with respective threshold values of \SI{330}{kW} and \SI{350}{kW} for the overall power consumption (see \cref{fig:energy_hysteresis}).
The model was simulated for a full year to capture seasonal variations in both renewable energy availability and power demand. 

\begin{algorithm}
    \uIf{$\mathrm{Temperature} \leq \SI{16}{\celsius}$}{
        $\mathrm{PowerOutput} \gets 0.4 \; \mathrm{PowerReference}$
    }
    \Else{
        $\mathrm{PowerOutput} \gets 0$
    }
    \caption{CHP controller}
    \label{alg:chp_controller}
    \vspace{-0.1cm}
\end{algorithm}

\begin{algorithm}
    \uIf{$\mathrm{Discharge}$}{
        $\mathrm{PowerOutput} \gets -200\,000$
    }
    \uElseIf{$\mathrm{PowerReference} \leq 200\,000$}{
        $\mathrm{PowerOutput} \gets 100\,000$
    }
    \Else{
        $\mathrm{PowerOutput} \gets 100\,000$
    }
    \caption{Battery controller}
    \label{alg:battery_controller}
\end{algorithm}

\begin{figure}
    \centering
    \includegraphics[width=0.5\textwidth]{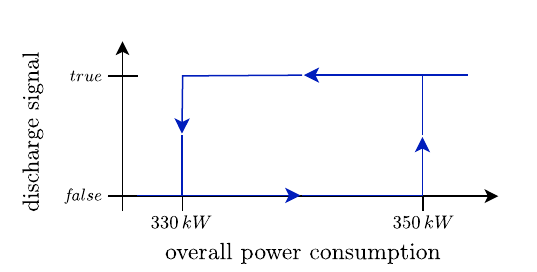}
    \caption{\textbf{Hysteresis function of the battery controller.} Depending on the overall power consumption level $u$, a discharge signal $y$ is sent to the controller.}
    \label{fig:energy_hysteresis}
\end{figure}

The goal of the PID analysis is to find the environmental variables which best predict the net energy transfer to the power grid $P_\mathrm{Grid}$. Specifically, we performed analyses between the net energy transfer as the target variable and used the ambient temperature, $T_\mathrm{A}$, wind speed, $v_\mathrm{Wind}$, and time of day, $t_\mathrm{Day}$, as sources. A full PID between these three source variables reveals how the predictive capabilities of the environmental variables is distributed among the individual factors.

\subsection{Results of experiments}

\begin{table}
\vspace{-0.35cm}
    \begin{ruledtabular}
    \caption{\textbf{Most information that the ambient temperature ($S_1=T_A$) and the wind speed ($S_2=v_\mathrm{Wind}$) hold about the net energy transfer ($P_\mathrm{Grid}$) is held synergistically.} The PID is calculated considering the target variable to be the net energy transfer to the electricity grid. The estimator was applied to $N=10^5$ data points.}
            \begin{tabular}{cc}
        &  Size of information atom
        \\
        \hline
        \cellcolor{syn}$\Pi_\mathrm{syn}$      & $1.363    \, \mathrm{bits}$
        \\
        \cellcolor{unq1}$\Pi_{\mathrm{unq},1}$ & $0.296    \, \mathrm{bits}$
        \\
        \cellcolor{unq2}$\Pi_{\mathrm{unq},2}$ & $0.034    \, \mathrm{bits}$
        \\
        \cellcolor{redu}$\Pi_\mathrm{red}$    & $0.630 \, \mathrm{bits}$
        \\
        $I(T:S_1, S_2)$ & $2.323 \, \mathrm{bits}$
        \label{tab:honda_2var}
    \end{tabular}
    \end{ruledtabular}
\end{table}

\subsubsection{Bivariate PID results}
As a first step, we analyzed how well the net energy transfer $P_\mathrm{Grid}$ is predictable from the ambient temperature $T_A$ and the wind speed $v_\mathrm{Wind}$. The results of this bivariate PID are summarized in \cref{tab:honda_2var}. In total, $2.323 \, \mathrm{bits}$ of information about $P_\mathrm{Grid}$ can be obtained by observing $T_A$ and $v_\mathrm{Wind}$. Intuitively, if all values for $P_\mathrm{Grid}$ were equally likely, one bit of information would be equivalent to restricting the range of possible values for the energy transfer by half after observation of the environmental variables. Accordingly, the observed $2.323 \, \mathrm{bits}$ would then equate to a reduction of range of the possible outcomes by a factor of about $1/5$.

Most of this information, namely $1.363 \, \mathrm{bits}$ are carried synergistically between the two environmental variables, meaning that this information is not accessible from any individual source alone. This implies that the net energy transfer is, to a large extent, predictable only from the interaction between the temperature and wind speed. Some of the information can also be determined redundantly from both sources ($\Pi_{\mathrm{red}} = 0.630 \, \mathrm{bits}$) and finally, some information can be gained from the ambient temperature alone ($\Pi_{\mathrm{unq},1} = 0.296 \, \mathrm{bits}$), while the wind speed conveys hardly any information exclusively by itself ($\Pi_{\mathrm{unq},2} = 0.034 \, \mathrm{bits}$). 

Note the strong similarity in values to the sum gate \cref{tab:analytical_examples}, which constructs the target as a net sum of its two source inputs. One can assume that this results from the dependency on the ambient temperature in the building power consumption and the CHP controller~\cref{alg:chp_controller} and the dependency of the energy generation of the wind turbine on the wind speed leading to these two source variables' effects would \emph{add up} together to inform about the net energy transfer to the external grid. However, the inequality in effects on other parts of the system such as potential temperature and wind speed dependency of for instance the photovoltaic system lead to a less balanced summation then was done in the equiweighted sum gate, and becomes conceptually more similar to a skewed sum gate as in \cref{tab:skewed_sum} in \cref{apx:skewed_sum}.

\subsubsection{Trivariate PID results}

\begin{figure}
    \centering
    \includegraphics{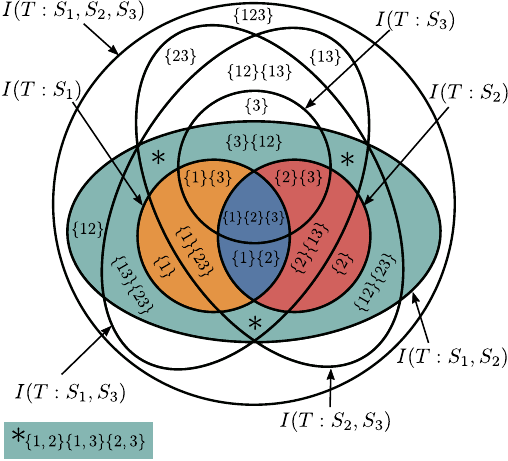}
    \caption{\textbf{By adding a third source variable $S_3$, the four PID atoms between the target $T$ and the sources $S_1$ and $S_2$ of \cref{fig:pid} are further subdivided and new ones are added, resulting in a total of $16$ atoms~\cite{williams2010}.} The individual atoms are identified by their antichains, e.g., $\{1\}\{2,3\}$ is the redundancy between source $S_1$ taken alone and sources $S_2$ and $S_3$ synergistically taken together.}
    \label{fig:trivariate-from-bivariate}
\end{figure}

A possible common cause for the correlation between the environmental variables and the net energy transfer from or to the power grid might be the time of day. To investigate this hypothesis, we augmented our analysis by incorporating the time of day as a third source variable in our PID analysis. By adding a third variable, the bivariate PID atoms themselves are dissected into smaller parts depending on their relation with the third variable (see \cref{fig:trivariate-from-bivariate}). For instance, the bivariate synergy $\Pi_\mathrm{syn}$ splits into five parts upon the introduction of the third variable, which are each characterized by a different antichain $\alpha$:

\begin{itemize}
    \item The part of the synergy between only the sources $S_1$ and $S_2$ which is \textbf{not} part of $S_3$ ($\alpha = \{1, 2\}$),
    \item the part of the synergy between $S_1$ and $S_2$ which is redundant with $S_3$ ($\alpha = \{1, 2\}\{3\}$),
    \item the part of the synergy between $S_1$ and $S_2$ which is redundant with the synergy between $S_1$ and $S_3$ ($\alpha = \{1, 2\}\{1, 3\}$),
    \item the part of the synergy between $S_1$ and $S_2$ which is redundant with the synergy between $S_2$ and $S_3$ ($\alpha = \{1, 2\}\{2, 3\}$) and finally
    \item the part of the synergy between $S_1$ and $S_2$ which is redundant with both the synergy between $S_1$ and $S_3$ and the synergy between $S_2$ and $S_3$ ($\alpha = \{1, 2\}\{1, 3\}\{2, 3\}$).
\end{itemize}

\begin{table*}
    \centering
    \caption{\textbf{The three source PID of the Honda energy management system between the ambient temperature ($S_1 = T_A$), the wind speed ($S_2=v_\mathrm{Wind}$), and the time of day ($S_3= t_\mathrm{Day}$) with the net energy transfer to an external grid ($T=P_\text{grid}$) as target variable reveals that most information that the sources carry about the target is redundant with the time of day.} The bivariate PID atoms from \cref{tab:honda_2var} are subdivided into trivariate atoms according to \cref{fig:trivariate-from-bivariate}, and the estimated information carried by each atom $\widehat{\Pi}(T:\alpha)$ as well as the cumulative atoms, i.e., the corresponding estimated redundant information $\widehat{I^\mathrm{sx}_\cap}(T:\alpha)$ are shown. The estimator was applied to the same $N=10^5$ data points as in \cref{tab:honda_2var}. The total mutual information is $I(T:S_1, S_2, S_3) = 3.092$ bits, agreeing with the top-most redundancy in the lattice. The bold rows each represent the most informative trivariate atom from the decomposition of the corresponding bivariate atom. Note that all bold rows contain an atom which is in some way a redundancy with source variable three.}
    \begin{ruledtabular}
    \begin{tabular}{cccc}
        &&\multicolumn{2}{c}{Estimation results}\\
        
        Antichain $\alpha$ of trivariate atom&Contained in bivariate atom & $\widehat{\Pi}(T:\alpha)$ & $\widehat{I^\mathrm{sx}_\cap}(T:\alpha)$\\
        \hline
        $\{1, 2, 3\}$ & - & $0.429 \, \mathrm{bits}$ & $3.092 \, \mathrm{bits}$ \\
        $\{1, 3\}$ & - & $0.688 \, \mathrm{bits}$ & $2.440 \, \mathrm{bits}$\\
        $\{2, 3\}$ & - &  $0.223 \, \mathrm{bits}$ & $1.971 \, \mathrm{bits}$\\
        $\{1, 3\} \{2, 3\}$ & - & $0.177 \, \mathrm{bits}$  & $1.742 \, \mathrm{bits}$\\
        $\{3\}$ & - & $-0.749 \, \mathrm{bits}$ & $1.239 \, \mathrm{bits}$\\
        
        $\{1, 2\}$ & \cellcolor{syn}$\Pi_\mathrm{syn}$ & $0.000 \, \mathrm{bits}$ & $2.330 \, \mathrm{bits}$\\
        $\{1, 2\} \{1, 3\}$ & \cellcolor{syn}$\Pi_\mathrm{syn}$ &  $0.013 \, \mathrm{bits}$  & $2.324 \, \mathrm{bits}$\\
        $\{1, 2\} \{2, 3\}$ & \cellcolor{syn}$\Pi_\mathrm{syn}$ & $0.006 \, \mathrm{bits}$  & $2.320 \, \mathrm{bits}$\\
        $\mathbf{\{3\} \{1, 2\}}$ & \cellcolor{syn}$\mathbf{\Pi_\textbf{syn}}$ & $\mathbf{1.053 \, \textbf{bits}}$ & $1.988 \, \mathrm{bits}$\\
        $\{1, 2\} \{1, 3\} \{2, 3\}$ & \cellcolor{syn}$\Pi_\mathrm{syn}$ &$0.29 \, \mathrm{bits}$ & $2.314 \, \mathrm{bits}$ \\

        $\{1\}$ & \cellcolor{unq1}$\Pi_{\mathrm{unq},1}$ &  $-0.003 \, \mathrm{bits}$ & $0.927 \, \mathrm{bits}$\\
        $\{1\} \{2, 3\}$ & \cellcolor{unq1}$\Pi_{\mathrm{unq},1}$ & $0.022 \, \mathrm{bits}$ & $0.930 \, \mathrm{bits}$ \\
        $\mathbf{\{1\} \{3\}}$ & \cellcolor{unq1}$\mathbf{\Pi_{\textbf{unq},1}}$ & $\mathbf{0.278 \, \textbf{bits}}$ & $0.905 \, \mathrm{bits}$\\

        $\{2\}$ & \cellcolor{unq2}$\Pi_{\mathrm{unq},2}$ & $0.000 \, \mathrm{bits}$ & $0.664 \, \mathrm{bits}$\\
        $\{2\} \{1, 3\}$ & \cellcolor{unq2}$\Pi_{\mathrm{unq},2}$ & $0.004 \, \mathrm{bits}$ & $0.664 \, \mathrm{bits}$ \\
        $\mathbf{\{2\} \{3\}}$ & \cellcolor{unq2}$\mathbf{\Pi_{\textbf{unq},2}}$ & $\mathbf{0.03} \, \textbf{bits}$  & $0.657 \, \mathrm{bits}$\\
	        
		$\{1\} \{2\}$ & \cellcolor{redu}$\Pi_\mathrm{red}$ & $0.003 \, \mathrm{bits}$ & $0.630 \, \mathrm{bits}$ \\
		$\mathbf{\{1\} \{2\} \{3\}}$ & \cellcolor{redu}$\mathbf{\Pi_\textbf{red}}$ &  $\mathbf{0.627 \, \textbf{bits}}$ & $0.627 \, \text{bits}$ \\
    \end{tabular}
    \end{ruledtabular}
    
    \label{tab:toy_example_3var}
    \vspace{-0.1cm}
\end{table*}

Examining the estimation results of the trivariate PID, as listed in \cref{tab:toy_example_3var} and visualized in \cref{fig:toy_example_3var_venn}, we can observe that the synergy between ambient temperature and wind speed primarily resides within the trivariate atom characterized by the antichain $\{1, 2\}\{3\}$. This implies that the bivariate synergy between $T_A$ and $v_\mathrm{Wind}$ largely duplicates the information conveyed by the time of day. A similar trend emerges when considering bivariate redundancy $\Pi_\mathrm{red}$, which also predominantly overlaps with the time of day, i.e., resides mainly in the trivariate atom $\{1\}\{2\}\{3\}$, and also the bivariate unique information of the ambient temperature $\Pi_{\mathrm{unq},1}$, which mostly resides in the redundant atom $\{1\}\{3\}$. These findings substantiate our hypothesis that weather variables are primarily redundant with the time of day. One possible explanation for this is that the weather variables primarily act as intermediaries, conveying information that the time of day carries about the net energy transfer.

Nevertheless, there are also other PID atoms which contribute non-negligibly to the total mutual information. In particular, some information is carried synergistically between the ambient temperature and the time of day ($\{1, 3\}$) or, to a lesser extent, between the synergy between all variables ($\{1, 2, 3\}$) and the synergy between wind speed and the time of day ($\{2, 3\}$). Yet other relevant contributions come from the redundancy between the ambient temperature and the time of day ($\{1\}\{3\}$) and the redundancy between all three two-variable synergies ($\{1, 2\}\{1, 3\}\{2, 3\}$).

Further, note that the information that the time of day provides uniquely and independent of ambient temperature and wind speed is on average negative, i.e., misinformative about the total power draw. Average misinformation can arise in $I^\mathrm{sx}_\cap$ since the operational interpretation, which has been put forward in \citet{makkeh} and carries over to the continuous case, assumes a memoryless agent and interprets the average information-theoretic quantities as ensemble averages. For a hypothetical agent incapable of learning from past events, this means that gaining only the purely unique information of the time of day diminishes the agent's accuracy in predicting the net energy transfer.

\begin{figure}
    \centering
    \includegraphics[width = \columnwidth]{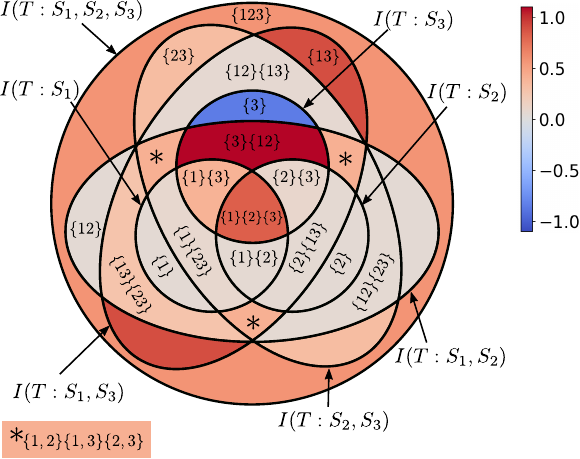}
    \caption{\textbf{Most of the information in $S_1$ (ambient temperature) and $S_2$ (wind speed) is redundant with $S_3$ (time of day).} Visualization of the results of \cref{tab:toy_example_3var} in the mereological diagram from \cref{fig:trivariate-from-bivariate}. Here, red encodes positive, i.e., informative, information contributions carried by the associated information atom, while blue encodes negative, i.e., misinformative contributions in bits.}
    \label{fig:toy_example_3var_venn}
\end{figure}

\section{Discussion}

This work presents a tractable analytical definition of continuous PID which is based on the intuitive ideas of the discrete $I^\mathrm{sx}_\cap$ measure~\cite{makkeh} and the measure-theoretical existence proofs by \citet{schick-poland}. It is in a local form by definition, is applicable to arbitrary many source variables and differentiable. We exemplify the intuition of our definition on four toy examples and observe that the measure gives mostly qualitatively comparable results to the the discrete case when computed on continuous versions of certain logic gates. Furthermore, we provide a nearest-neighbor based estimator and show that it converges to the numerically obtained values on the before-mentioned examples. Lastly, we demonstrate how this estimator can be used in a practical example by using it to uncover variable dependencies in a simulated energy management system.

The proposed measure fills a gap not addressed by existing continuous PID measures. The space of PID measures can be subdivided into mainly two groups: Those based on optimization of an auxiliary variable with respect to some ordering criterion \cite{kolchinsky2022novel, griffith2014intersection, griffith2015quantifying} as explained in the framework of \citet{kolchinsky2022novel}, and those which have a \emph{pointwise} definition, i.e.,  which can be defined for an individual realization in a self-contained functional way dependent solely on the probability densities at that realization. \citet{kolchinsky2022novel} highlights that for ordering-based PID measures, the inclusion-exclusion-rule for defining PID atoms is inconsistent and thus cannot be applied. However, being defined for local realizations, our proposed measure is fully compatible to the mereological framework of \citet{gutknecht2021} and is therefore not affected by this criticism.

Furthermore, it is crucial to remark that mutual information as well as PID are measures of statistical dependence only, and therefore do not imply any causation in the interventional sense. Note that a causality decomposition into causal atoms must have a different structure from the partial information decomposition lattice as causality does not follow the same monotonicity criterion as redundancy~\cite{gutknecht2023information}. This is because multiple causal influences can hinder and prevent each other: Consider, for instance, two stones being thrown at a bottle. While each stone individually would hit and destroy the bottle, their trajectories may intersect in such a way that the stones themselves collide when thrown simultaneously, diverting each other enough to both miss the bottle. In this case, two causes which would have individually been sufficient to cause an event have cancelled each other, which is an interaction that is impossible when considering redundancy instead.

As outlined in \cref{sec:analytical}, the definition of continuous $I^\mathrm{sx}_\cap$ leaves open a choice of how to set the relative scale between variables that needs to be made in accordance with the setup in question. While such a choice marks a departure from the purely model-free nature of classical information theory, note that this is likely unavoidable in a PID context. For PID itself, choosing one of the multiple competing PID measures according to their operational interpretation similarly introduces some notion of ``proto-semantics'' to the analysis. Overall, we see the choice of relative scale not as a weakness but as a necessity for a continuous PID measure, since what it means to be in a neighborhood of a value might differ drastically between different variables.

This work builds on the definition of the discrete $I^\mathrm{sx}_\cap$ measure \cite{makkeh} and its measure-theoretic extension to continuous variables \cite{schick-poland} to introduce a tractable analytical definition and practical nearest-neighbor based estimator for a continuous $I^\mathrm{sx}_\cap$ measure. However, this definition turns out not to be invariant with respect to bijective transformations of the individual variables--at least not without prescribing a preprocessing scheme (see \cref{apx:preprocessing}). Yet, this invariance is a property central to the measure-theoretic derivation in \cite{schick-poland}. For this reason, the analytic definition presented here cannot be directly interpreted as a concretization of the latter but should rather be seen as a practical implementation rooted in the intuitive meaning of the $I^\mathrm{sx}_\cap$ measure.

Like in the discrete case, this definition of redundancy might yield negative information atoms. In the context originally introduced by \citet{makkeh}, namely for memory-less agents and interpreting the average in \cref{eq:discrete_bivariate_global_sxpid} as an ensemble average, these are easily interpreted as misinforming the agents of the ensemble on average. For time-series data, for which the averages are taken over time, however, negative atoms can be a distracting technicality hindering a straightforward interpretation of the results. In these cases, practitioners may opt to clamp the atoms to non-negative values and adjust the other atoms accordingly so that the consistency equations (Equations \eqref{eq:redundancy_from_atoms} for the general case) continue to hold.

In this paper, we propose a novel shared-exclusion based PID measure designed specifically for purely continuous variables. However, it is important to acknowledge that real-world systems often involve a combination of continuous and discrete source variables. Additionally, even certain individual variables may comprise both discrete and continuous parts. To address these scenarios, we introduce an analytical formulation in \cref{apx:mixed}, outlining an ansatz for handling mixed systems. While we defer the presentation of specific examples demonstrating its application and the development of a mixed continuous-discrete estimator to future research, this work lays the groundwork for a comprehensive understanding of information decomposition in systems with diverse variable types. 

Overall, with this paper we introduce a new analytical formulation for a shared-exclusion based continuous PID measure and provide an estimator that makes it applicable to practitioners from all sciences and technology.

\section*{Acknowledgements}
We would like to thank Anja Sturm for fruitful discussions and valuable feedback on this paper. We further want to thank everyone in the Wibral lab and the Honda Research Institute for their feedback and indirect contributions.

DE, KSP, AM and MW are employed at the Göttingen Campus Institute for Dynamics of Biological Networks (CIDBN). DE and MW were supported by funding from the Ministry for Science and Education of Lower Saxony and the Volkswagen Foundation through the ``Niedersächsisches Vorab'' under the program ``Big Data in den Lebenswissenschaften'' -- project ``Deep learning techniques for association studies of transcriptome and systems dynamics in tissue morphogenesis''. KSP, FL and PW are funded by the Honda Research Institute Europe GmbH.

\appendix
\crefalias{section}{appendix}

\section{On the shared-exclusions measure of redundancy}
\label{appendix:isx_visual}
The definition of bivariate $I^\mathrm{sx}_\cap$ in \cref{eq:discrete_bivariate_local_sxpid} can also be derived from shared exclusions of parts of the underlying abstract probability space $(\Omega, \mathbb{P})$, i.e., the parts of $(\Omega, \mathbb{P})$ that are rendered impossible by observation of either variable. To determine the precise subsets of the probability space to be excluded, we assume that a realization $\bm S = (S_1, S_2) = \bm s = (s_1, s_2)$ has been taken by the system and we have gained the information that either $S_1=s_1$ or $S_2 = s_2$ have materialized. Through this observation, the part of the probability space corresponding to any other realization $\bm s' = (s_1', s_2')$ is rendered impossible unless either $s_1' = s_1$ or $s_2' = s_2$. If neither of those equations holds, then $\bm s'$ would not fulfill the logical or-statement in \cref{eq:discrete_bivariate_local_sxpid} and would hence be impossible under the assumption (see \cref{fig:exclusion_prob_space}).

\begin{figure*}
    \centering
    \includegraphics[width=0.8\textwidth]{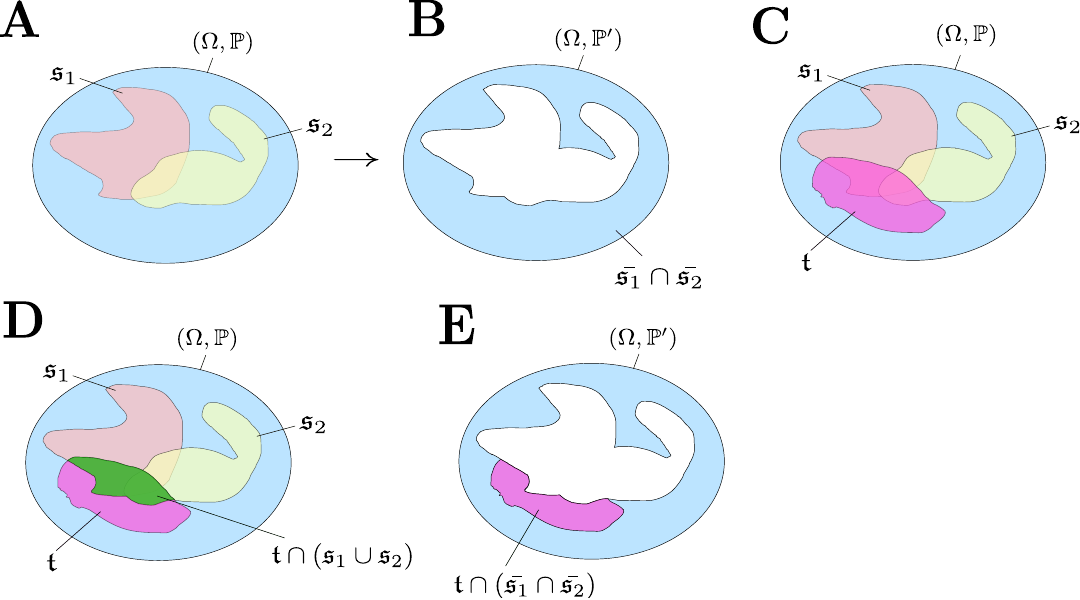}
    \caption{\textbf{The pointwise redundancy $i^{\mathrm {sx}}_\cap$ can be derived from principles of probability exclusion.} Local redundant information for two sources $S_1$ and $S_2$ with respect to target $T$ at a joint realization $(s_1, s_2, t)$. Each realization corresponds to one of the events in the probability space denoted by $\mathfrak{s}_1, \mathfrak{s}_2$ and $\mathfrak{t}$ in the probability space $(\Omega, \mathbb{P})$. Here we aim at assessing the benefit of knowing that either $S_1 = s_1$ or $S_2 = s_2$ in predicting $t$. In panel $\mathbf{A}$ the events $\mathfrak{s}_1$ and $\mathfrak{s}_2$ are displayed. In panel $\mathbf{B}$ there is the intersection of both complements from the events in panel $\mathbf{A}$, which is the area of probability mass which either event excludes. In panel $\mathbf{C}$ we consider a target variable and its event set $\mathfrak{t}$ in $\Omega$, while panel $\textbf{D}$ shows the area of interest, the green set represents the part of event $\mathfrak{t}$ that is included in also either $\mathfrak{s}_1$ or in $\mathfrak{s}_2$ which corresponds to the part of the probability of realization $t$, which is accessible when assuming that the logical statement $S_1 = s_1 \lor S_2 = s_2$ is true, or, correspondingly, $W_{\{1\},\{2\}} = 1$. While the green area can be measured via the probability measure directly, i.e., $\mathbb{P}(\mathfrak{t} \cap (\mathfrak{s}_1 \cup \mathfrak{s}_2))$, we can similarly represent it using the shared exclusion overlapping with $\mathfrak{t}$ as in panel $\mathbf{E}$, by $\mathbb{P}'(\mathfrak{t}) - \mathbb{P}'(\mathfrak{t} \cap (\bar{\mathfrak{s}_1} \cap \bar{\mathfrak{s}_2})) = \frac{\mathbb{P}(\mathfrak{t}) - \mathbb{P}(\mathfrak{t} \cap (\bar{\mathfrak{s}_1} \cap \bar{\mathfrak{s}_2}))}{1 - \mathbb{P}(\mathfrak{s}_1 \cup \mathfrak{s}_2)}$. This green area then has to be compared to the regular probability of obtaining the realization $t$ without knowledge of any of the sources, $\mathbb{P}(\mathfrak{t})$.}
    \label{fig:exclusion_prob_space}
\end{figure*}

When parts of the probability space are rendered impossible by observing a certain event, the resulting posterior probability measure $\mathbb{P}$ needs to be normalized on $\Omega$, yielding an adapted probability measure $\mathbb{P}'$, known as the conditional probability measure. The visualization \cref{fig:exclusion_prob_space} shows the exclusion of specific parts of the probability space leading to descriptions of the information that can be assessed from either $S_1$ or $S_2$ locally about $T$.

Similarly, this concept can then be generalized to collections $\alpha = \{\bm a_j\}_j = \{\{s_{ij}\}_i\}_j$ of three or more source variables by excluding the part of the probability space that is redundantly excluded by all joint collections of random variables $a_j$. For instance, in a quadrivariate (i.e., four source variable) system, the redundancy between the collections $\{S_1, S_3\}$ and $\{S_2, S_4\}$ is defined by the joint exclusions between the collections, which is given by all realizations $\bm s' = (s'_1, s'_2, s'_3, s'_4)$ for which neither $(s'_1, s'_3) = (s_1, s_3)$ nor $(s'_2, s'_4) = (s_2, s_4)$ holds. This generalization to collections of subsets of the source variables turns out to be of the form of local mutual information with respect to an auxiliary random variable 
$$\mathcal{W}_{\Phi}(\bm s) = \bigvee\limits_{a \in \mathcal{P}(\mathbf{S}):\Phi(a)=1} \; \bigwedge\limits_{i \in a} (S_i = s_i)$$
and reverts to a classical mutual information quantity for \emph{self-redundancies}, i.e., $\alpha = \{\bm a\}$. Further, its dependence on elements of the power set of the index set of sources (set of all collections of sources) implies that the $I^\mathrm{sx}_\cap$ measure of redundant information gives rise to a full lattice of PID atoms as described in \cref{sec:nd_pid_lattice}. Note also that this measure, as a composition of differentiable functions, is differentiable itself with respect to minor perturbations of the underlying probability mass function. Additionally, it fulfills a target chain rule~\cite{makkeh}.

\section{Construction of densities of logical statements} \label{appendix:derivation_density_logical_statements}

In this appendix section we argue for the explicit form of the quasi-densities appearing in $I^\mathrm{sx}_\cap$, which stem from a logical statement allowing also for disjunctions, and not only conjunctions as is usual in traditional probability theory.

Let us recall that when measuring a union of two sets, by the law of inclusion-exclusion, we can express the result as the following sum:
\begin{align*}
    \mathbb{P}^{S_1, S_2}(B \cup C) = \mathbb{P}^{S_1, S_2}(B) + \mathbb{P}^{S_1, S_2}(C) - \mathbb{P}^{S_1, S_2}(B \cap C) \, . 
\end{align*}
Similarly to \cref{appendix:isx_visual}, a union can be utilized to represent the logical disjunction $S_1 \in B \lor S_2 \in C$ via measuring the set $B \times E_{S_2} \cup E_{S_1} \times C$ where $E_{S_i}$ denotes the respective space of possible realizations for $S_i$. Measuring $B \times E_{S_2} \cup E_{S_1} \times C$ then yields
\begin{gather*}
\mathbb{P}^{S_1, S_2}(B \times E_{S_2} \cup E_{S_1} \times C)\\
= \mathbb{P}^{S_1}(B) + \mathbb{P}^{S_2}(C) - \mathbb{P}^{S_1, S_2}(B \cap C) \, ,    
\end{gather*}
i.e., by using the marginal measures, just as it is done in discrete settings.
Suppose now that we know that the statement $S_1 \in B \lor S_2 \in C$ is true. How would we measure the probability of additional events under this presumption? Assuming the event $S_1 \in B \lor S_2 \in C$ is true means that we need to use a conditional probability measure to measure additional events $A \in E_{S_1} \times E_{S_2}$.
Thus, to obtain a probability measure under the assumption that $S_1 \in B \lor S_2 \in C$ is true, we can write

\begin{gather*}
    \mathbb{P}^{(S_1, S_2)}_{S_1 \in B \lor S_2 \in C}(A) := \mathbb{P}^{S_1}(A \cap B) + \mathbb{P}^{S_2}(A \cap C)\\
    - \mathbb{P}^{(S_1, S_2)}(A \cap (B \times C)).
\end{gather*}

Applying this to infinitesimally small sets $B = \{s_1\}$ and $C = \{s_2\}$, we can achieve the same form by adding slices of conditional measures with respect to the constraints set by the disjunction. We define 
\begin{gather*}
    \mathbb{P}^{(S_1, S_2)}_{(S_1 = s_1) \lor (S_2 = s_2)}(A) := \mathbb{P}^{S_1}(\{x\vert (x, y) \in A, y = s_2\})\\
    + \mathbb{P}^{S_2}(\{y \vert (x, y) \in A, x = s_1\})
    - \mathbb{P}^{(S_1, S_2)}(\{(s_1, s_2)\} \cap A)
\end{gather*}
for an event $A$. 

Note that when taking the limit to infinitesimal sets, a choice needs to be made regarding the relative rate of convergence of the two variables. Choosing to interpret the two variables as being on equal scales, i.e., using the same $\epsilon$ to search in the respective marginal neighborhoods, this uniquely defines the measure as the only one fulfilling
\begin{widetext}
\begin{align*}
    \norm{\mathbb{P}^{(S_1, S_2)}_{(S_1 \in B_{s_1}(\epsilon)) \lor (S_2 \in B_{s_2}(\epsilon))}(A) - \epsilon \cdot \mathbb{P}^{(S_1, S_2)}_{(S_1 \in \{s_1\}) \lor (S_2 \in  \{s_2\})}(A)} \leq C_\epsilon \cdot \epsilon \overset{\epsilon \searrow 0}{\to} 0 \quad \forall A \in \mathcal{E} \, ,
\end{align*}
\end{widetext}
with $C_\epsilon \overset{\epsilon \searrow 0}{\to} C$, where $C$ is a constant. This can easily be seen when using the triangle inequality and recalling the form of a conditional probability measure, i.e.,
\begin{widetext}
\begin{align*}
        \mathbb{P}^{S_1, S_2}(\{(x, y) \in A \vert y \in B_{s_2}(\epsilon)\}) &= \underbrace{\frac{\mathbb{P}^{S_1, S_2}(\{(x, y) \in A \vert y \in B_{s_2}(\epsilon)\})}{\mathbb{P}^{S_2}(B_{s_2}(\epsilon))}}_{\overset{\epsilon \searrow 0}{\to} \mathbb{P}^{S_1}(\{(x, y) \in A \vert y = s_2\})} \underbrace{\frac{\mathbb{P}^{S_2}(B_{s_2}(\epsilon))}{\lambda_{S_2}(B_{s_2}(\epsilon))}}_{\overset{\epsilon \searrow 0}{\to} f_{S_2}(s_2)} \underbrace{\lambda_{S_2}(B_{s_2}(\epsilon))}_{2\epsilon} \\
        &\quad \overset{\epsilon \ll 1}{\sim} 2 \epsilon \mathbb{P}^{S_1}(\{(x, y) \in A \vert y = s_2\}) f_{S_2}(s_2) \, .
    \end{align*}
    \end{widetext}
    This idea is further visualized in \cref{figure_limit_continuity_measure}.
    \begin{figure*}
        \centering
        \includegraphics[width=\textwidth]{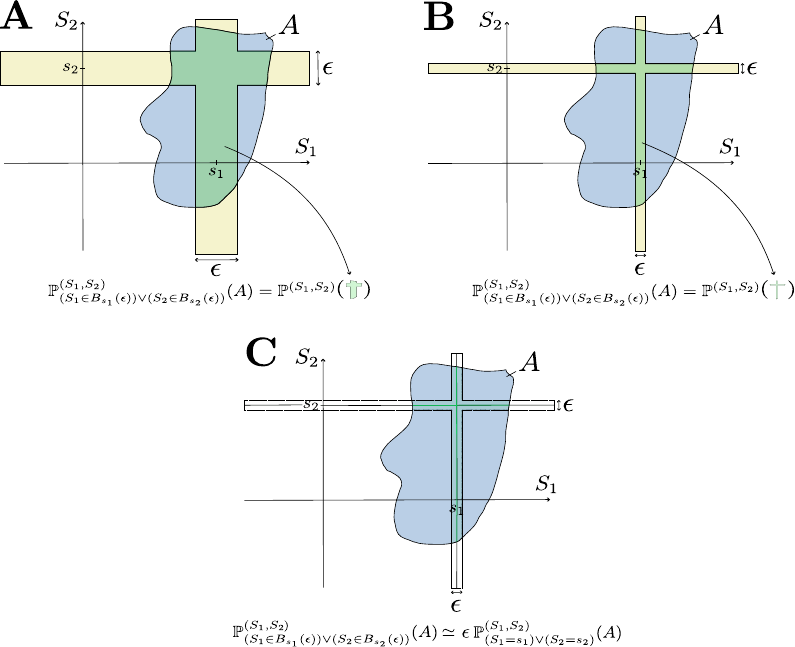}
        \caption{\textbf{We choose our localized conditional probability measure such that its extension for small neighbourhoods naturally results in the global probability measure.} In panel $\textbf{A}$ the intuitive meaning of the measure $\mathbb{P}^{(S_1, S_2)}_{(S_1 \in B_{s_1}(\epsilon)) \lor (S_2 \in B_{s_2}(\epsilon))}(A)$ for a set $A$ is illustrated. In panel $\textbf{B}$ one sees the consequences of $\epsilon$ decreasing, leading to $\mathbb{P}^{(S_1, S_2)}_{(S_1 \in B_{s_1}(\epsilon)) \lor (S_2 \in B_{s_2}(\epsilon))}(A)$ approaching the conditional measure of our choice, $\mathbb{P}^{(S_1, S_2)}_{(S_1 \in \{s_1\}) \lor (S_2 \in  \{s_2\})}(A)$ multiplied with the width of the neighborhood, $\epsilon$, as shown in panel $\textbf{C}$.}
        \label{figure_limit_continuity_measure}
    \end{figure*}
    
    Evidently, this ansatz fails if any of the marginal densities $f_{S_i}$ is unbounded at $s_i$.
    The choice of this measure could also be argued for as a consequence of the simple function approximation, found for instance in \cite{Rudin87}, thrm. 7.10.
    
To arrive at a probability for an infinitesimal event, compare the above to the construction of the usual Lebesgue measure $\lambda_{S}$ measuring the maximal size of nicely shrinking balls $B_i = B_{(s_1, s_2)}(r_i)$ \cite{Rudin87}. Note that the maximal size of a $B_i$ is the diameter of the ball, $2r_i$, and in the continuous case, the last term in the numerator vanishes.

Then
\begin{align*}
    &f_{S_1 \lor S_2}(s_1, s_2)\\
    =& \lim\limits_{B_i \searrow \{(s_1, s_2)\}} \frac{\mathbb{P}^{(S_1, S_2)}_{(S_1 = s_1) \lor (S_2 = s_2)}(B_{(s_1, s_2)}(r_i))}{\lambda_{S}(B_{(s_1, s_2)}(r_i))} \\
    =& \lim\limits_{r_i \searrow 0} \frac{\mathbb{P}^{S_1}(\{(x, y) \in B_i \vert y = s_2\})}{2r_i} \\
    &+ \frac{\mathbb{P}^{S_2}(\{(x, y) \in B_i \vert x = s_1\})}{2r_i}\\
    =& \lim\limits_{r_i \searrow 0} \frac{\mathbb{P}^{S_1}(\{s_1 \pm r_i\}) + \mathbb{P}^{S_2}(\{s_2 \pm r_i\})}{2r_i}\\
    =& f_{S_1}(s_1) + f_{S_2}(s_2) \, ,
\end{align*}
where $\lambda_{S_1}(\{(x, y) \in  B_i \vert y = s_2\}) + \lambda_{S_2}(\{(x, y) \in B_i \vert x = s_1\}$ is equal to $2r_i$ since the measures $\lambda_{S_i}$ measure the diameter of the uniform ball centered around $(s_1, s_2)$.
    
This method is readily generalized to $n$ sources in the exact same way it worked for two, with the one exception being that conjunctions are treated as one variable. For example, a logical statement $\alpha(s_1, s_2, s_3) = \{S_1 = s_1\} \land \{S_2 = s_2\} \lor\{ S_3 = s_3\}$ is treated as a disjunction of a single two dimensional statement and a single one dimensional statement.

\begin{definition} \label{definition_1}

Let for a generic logical statement $\alpha$ of length $l(\alpha)$ in its disjunctive normal form $\alpha = \bigvee_{j \in J} a_j$, $I_d \subset J$ be the subset of purely discrete statements. Then we define the density of $\alpha$ as

\begin{widetext}
\begin{align*}
    \frac{\delta \mathbb{P}_\lor}{ \delta \lambda_\lor}(s_1, \ldots s_n) = f_\alpha(x_1, \ldots x_n) := \sum\limits_{k=1}^{|I_d| + 1} (-1)^{k+1} \left(\sum\limits_{\substack{m_1 < m_2 < \ldots < m_{k-1} \in I_d \\ m_1 \neq m_2 \neq \ldots \neq m_{k-1} \neq j \in J}} f_{\bigcup\limits_{q = 1}^{k-1} a_{m_q} \cup a_j}\left(\bigcup\limits_{q = 1}^{k-1} x_{a_{m_q}} \cup x_{a_j}\right) \right) \, .
\end{align*}
\end{widetext}

\end{definition}

\begin{table*}[!]
    \centering
    \caption{\textbf{Different preprocessing schemes for continuous PID gates produce qualitatively similar results on the redundant and unique gate, while a copula pretransformation results in different results for the copy and sum gates.} The results have been computed using a Monte-Carlo integration analogous to \cref{tab:analytical_examples}.}
    \begin{ruledtabular}
    \begin{tabular}{rcccc}
        &  Redundant Gate & Copy Gate & Unique Gate & Sum\\
        \hline
        No transformation\\
        \hline
        $\Pi_\mathrm{syn}$ & $0.000 \, \mathrm{bits}$ & $6.644 \, \mathrm{bits}$ & $5.525  \, \mathrm{bits}$ & $6.497  \, \mathrm{bits}$\\
        $\Pi_\mathrm{unq,1}$ & $0.000    \, \mathrm{bits}$ & $0.000 \, \mathrm{bits}$ & $1.119 \, \mathrm{bits}$ & $0.147  \, \mathrm{bits}$\\
        $\Pi_\mathrm{unq,2}$ & $0.000    \, \mathrm{bits}$ & $0.000 \, \mathrm{bits}$ & $-5.525 \, \mathrm{bits}$ & $0.147  \, \mathrm{bits}$\\
        $\Pi_\mathrm{red}$ & $6.644 \, \mathrm{bits}$ & $6.644 \, \mathrm{bits}$ & $5.525   \, \mathrm{bits}$ & $0.353  \, \mathrm{bits}$\\
        $I(T:S_1, S_2)$ & $6.644 \, \mathrm{bits}$ &  $13.288 \, \mathrm{bits}$ & $6.644 \, \mathrm{bits}$ & $7.144  \, \mathrm{bits}$\\
        \hline
        Standardization\\
        \hline
        $\Pi_\mathrm{syn}$ & $0.000 \, \mathrm{bits}$ & $6.644 \, \mathrm{bits}$ & $5.524  \, \mathrm{bits}$ & $6.497  \, \mathrm{bits}$\\
        $\Pi_\mathrm{unq,1}$ & $0.000    \, \mathrm{bits}$ & $0.000 \, \mathrm{bits}$ & $1.119 \, \mathrm{bits}$ & $0.147  \, \mathrm{bits}$\\
        $\Pi_\mathrm{unq,2}$ & $0.000    \, \mathrm{bits}$ & $0.000 \, \mathrm{bits}$ & $-5.524 \, \mathrm{bits}$ & $0.147  \, \mathrm{bits}$\\
        $\Pi_\mathrm{red}$ & $6.644 \, \mathrm{bits}$ & $6.644 \, \mathrm{bits}$ & $5.524   \, \mathrm{bits}$ & $0.353  \, \mathrm{bits}$\\
        $I(T:S_1, S_2)$ & $6.644 \, \mathrm{bits}$ &  $13.288 \, \mathrm{bits}$ & $6.644 \, \mathrm{bits}$ & $7.144  \, \mathrm{bits}$\\
        \hline
        Copula transformation\\
        \hline
        $\Pi_\mathrm{syn}$ & $0.000 \, \mathrm{bits}$ & $6.922 \, \mathrm{bits}$ & $5.671  \, \mathrm{bits}$ & $6.791  \, \mathrm{bits}$\\
        $\Pi_\mathrm{unq,1}$ & $0.000    \, \mathrm{bits}$ & $-0.279 \, \mathrm{bits}$ & $0.973 \, \mathrm{bits}$ & $-0.147  \, \mathrm{bits}$\\
        $\Pi_\mathrm{unq,2}$ & $0.000    \, \mathrm{bits}$ & $-0.279 \, \mathrm{bits}$ & $-5.671 \, \mathrm{bits}$ & $-0.147  \, \mathrm{bits}$\\
        $\Pi_\mathrm{red}$ & $6.644 \, \mathrm{bits}$ & $6.922 \, \mathrm{bits}$ & $5.671   \, \mathrm{bits}$ & $0.647  \, \mathrm{bits}$\\
        $I(T:S_1, S_2)$ & $6.644 \, \mathrm{bits}$ &  $13.288 \, \mathrm{bits}$ & $6.644 \, \mathrm{bits}$ & $7.144  \, \mathrm{bits}$\\
    \end{tabular}
    \end{ruledtabular}

    \label{tab:preprocessing_gates}
\end{table*}

\begin{example}
\begin{enumerate}
    \item Take the example $\alpha = (X_1 \land X_2) \lor (X_3) \lor (X_4 \land X_5)$, where $X_1, X_2$ and $X_4$ are discrete, and $X_3$ and $X_5$ are continuous, such that $J=\{1, 2, 3\}$ and $I_d = \{1\}$. The resulting quasi-density is
\begin{gather*}
    f_\alpha(x_1, x_2, x_3, x_4, x_5)\\
    = \underbrace{p_{X_1 X_2}(x_1, x_2) +f_{x_3}(x_3) + f_{X_4 X_5}(x_4, x_5)}_{k=1}\\
    - \quad \left(\underbrace{f_{X_1 X_2 X_3}(x_1, x_2, x_3) + f_{X_1 X_2 X_4 X_5}(x_1, x_2, x_4, x_5) }_{k=2}\right).
\end{gather*}
    \item Choosing $\alpha = (X_1 \land X_2) \lor (X_1 \land X_3)$ with $X_1, X_3$ discrete and $X_2$ continuous on the other hand, leads to
    
\begin{gather*}f_\alpha(x_1, x_2, x_3)\\
= f_{X_1 X_2}(x_1, x_2) + p_{X_1X_3}(x_1, x_3)
- f_{X_1 X_2 X_3}(x_1, x_2, x_3). 
\end{gather*}
\end{enumerate}

\end{example}

The above derived form of the density of an OR statement readily only works for purely continuous measures, as for purely discrete ones; $\epsilon^i \to 1$ and thus the last term does not vanish, yielding the expression commonly derived via the inclusion-exclusion rule, that is,
\begin{align*}
    \frac{\mathbb{P}_\lor(A_i, B_i)}{\epsilon^i} \to p_X(x_A) + p_Y(y_B) - p_{XY}(x_A, y_B) \, .
\end{align*}

Thus those corner cases of purely discrete and purely continuous systems can be understood as they have comparable $\epsilon^i_A = \epsilon^i_B$, but if one considers the case of a system consisting of both types of variables, one cannot fulfill this requirement. 
To circumvent this nuisance while still measuring the maximal possible amount of probability, we define the density of an OR-statement in the measure-theoretic case to be
\begin{widetext}
    
\begin{gather*}
    \frac{\delta \mathbb{P}_\lor}{ \delta \lambda_\lor}(x_A, y_B)\\
    := \lim\limits_{i \to \infty} \frac{\mathbb{P}_X(\Delta A^i)}{\epsilon_A^i} + \frac{\mathbb{P}_Y(\Delta B^i)}{\epsilon_B^i} 
    - \frac{\mathbb{P}_{XY}(\Delta A^i \times B_{i-1}) + \mathbb{P}_{XY} ( A_{i-1} \times \Delta B^i) - \mathbb{P}_{XY}(\Delta A^i \times \Delta B^i)}{\epsilon_A^i \epsilon_B^i} \max\left(\epsilon_A^i, \epsilon_B^i\right)
\end{gather*}
\end{widetext}
with $\Delta A^i = A_i\setminus A_{i+1}$, and the same for $\Delta B^i = B_i\setminus B_{i+1}$. With this definition, the special cases of purely discrete and purely continuous arise again, and the mixed case yields sums of probability masses and densities of the same order of magnitude with respect to the underlying space vanishing, i.e., in the case where $X$ is continuous and $Y$ is discrete, we find 
\begin{align*}
    \frac{\delta \mathbb{P}_\lor}{ \delta \lambda_\lor}(x_A, y_B) = f_X(x_A) + p_Y(y_B) - f_{XY}(x_A, y_B) \, .
\end{align*}

Requiring the same to hold for higher order statements $\alpha$, this can easily be generalized utilizing the inclusion-exclusion law, leading to \cref{definition_1}.

\section{Variable preprocessing}
\label{apx:preprocessing}

This appendix chapter discusses the effect that different preprocessing schemes have on the results of the introduced PID definition. Different to the mutual information and the measure-theoretic existence proofs by \cite{schick-poland}, the analytic definition for redundancy suggested in this paper is not invariant under isomorphic mappings of individual source variables. At the core, the reason for this fact is that in the continuous case, the meaning of a logical disjunction turns out to be contingent on the relative scale of the two variables, since densities are defined as the limits of neighbourhoods, and there is no canonical way to compare neighborhoods between two different variables to make them shrink at the same pace.

\cref{definition_1} has been made under the assumption that the two variables are on a comparable scale. While this holds true for many scientific questions in which the variables in question are distributed equivalently, in other scenarios a preprocessing step might become necessary to compare the variables in a way adequate for the research question. \cref{tab:preprocessing_gates} shows how the PID results of the logic gate examples in \cref{tab:analytical_examples} change when different preprocessing schemes are applied. First, note that preprocessing does not impact the mutual information quantities. While there are no qualitative differences between the preprocessing schemes for the redundant and the unique gate, the copula transformation, i.e. changing the random variables to their own cumulative density making them uniform on the unit interval, introduces some unique information in the copy gate and negative unique information in the sum gate. These results show that while different preprocessing schemes might often lead to comparable results, it is nevertheless important to choose a preprocessing scheme that matches the scientific question to yield reliable and interpretable results.

Note, also, that specifically the preprocessing based on standardization makes no difference compared to no transformation for the given examples. This, however, is not true in general but only if the two source random variables already have the same standard deviation, as is the case for all examples shown here.

\section{Proof of the differentiability of $I^\mathrm{sx}_\cap$}
\label{apx:differentiability}
The analytical definition of the redundancy (see \cref{eq:bivariate_local_sxpid}) allows us to prove its differentiability and even smoothness for both the local and global measure. 

\begin{theorem}
    The local continuous measure of shared information

    \begin{align*}
       i^\mathrm{sx}_\cap[f](t, \{s_1\}\{s_2\}) 
       &= \log_2 \left[\frac{f_{T| S_1 \lor S_2}(t| s_1 \lor s_2)}{f_T(t)}\right]\\
       &=\log_2 \left[\frac{f_{TS_1}(t, s_1) + f_{TS_2}(t, s_2)}{f_T(t)(f_{S_1}(s_1) + f_{S_2}(s_2))}\right]
    \end{align*}
    
    varies smoothly with respect to changes of the underlying joint probability density $f_{TS_1S_2}$. Moreover, for more than two source variables, $i_\cap^{\mathrm{sx}}[f_{TS_1S_2}](t:\alpha)$ is smooth for arbitrary antichains $\alpha$.
\end{theorem}
\begin{proof}
    Let $f_{TS_1S_2}$ be a density function and $(t, s_1, s_2)$ a point in state space such that $f_{TS_1S_2}(t, s_1, s_2) > 0$. Now, let $g_{TS_1S_2}$ be a function 
    such that $g_{TS_1S_2}(t, s_1, s_2) \neq 0$.
    
    Then we construct the derivative $D(i_\cap^{\mathrm{sx}}[f])g$ in the direction of $g$ by considering the limit of a shifted density $f_{TS_1S_2} \mapsto f_{TS_1S_2} + \epsilon \cdot g_{TS_1S_2}$ with $\epsilon > 0$, yielding
    \begin{widetext}
    \begin{align*}
        D(i_\cap^{\mathrm{sx}}[f])g = &\lim_{\epsilon \to 0} \frac{i_\cap^{\mathrm{sx}}[f_{TS_1S_2}+ \epsilon \cdot g_{TS_1S_2}](t:s_1 \lor s_2) - i_\cap^{\mathrm{sx}}[f_{TS_1S_2}](t:s_1 \lor s_2)}{\epsilon} \\
        = &\lim_{\epsilon \to 0} \frac{\log\left[\frac{f_{TS_1}(t, s_1) + f_{TS_2}(t, s_2) + \epsilon (g_{TS_1}(t, s_1) + g_{TS_2}(t, s_2))}{(f_T(t)+\epsilon g_T(t))(f_{S_1}(s_1)+f_{S_2}(s_2)+\epsilon(g_{S_1}(s_1)+g_{S_2}(s_2))}\right] - \log\left[\frac{f_{TS_1}(t, s_1) + f_{TS_2}(t, s_2)}{f_T(t)(f_{S_1}(s_1) + f_{S_2}(s_2))}\right] }{\epsilon} \\
        = &\lim_{\epsilon \to 0} \frac{\log_2 \left[1 + \epsilon \frac{g_{TS_1}(t,s_1) + g_{TS2}(t,s_2)}{f_{TS_1}(t,s_1) + f_{TS_2}(t,s_2)}\right] - \log_2 \left[1 + \epsilon \frac{g_{S_1}(s_1)+g_{S_2}(s_2)}{f_{S_1}(s_1)+f_{S_2}(s_2)}\right] - \log_2 \left[1 + \epsilon \frac{g_T(t)}{f_T(t)}\right]}{\epsilon}\, .
    \end{align*}
    \end{widetext}
    At this point, we  use the Taylor expansion $\ln[1+x] = \sum\limits_{n=1}^\infty \frac{(-1)^{n+1}x^n}{n} = x + \mathcal{O}(x^2)$ to determine the final form of the derivative
    \begin{widetext}
    \begin{align*}
        D(i_\cap^{\mathrm{sx}}[f])g &= \frac{g_{TS_1}(t,s_1) + g_{TS_2}(t,s_2)}{f_{TS_1}(t,s_1) + f_{TS_2}(t,s_2)} - \frac{g_{S_1}(s_1)+g_{S_2}(s_2)}{f_{S_1}(s_1)+f_{S_2}(s_2)} - \frac{g_T(t)}{f_T(t)} \, ,
    \end{align*}
    \end{widetext}
    which is finite by assumption.\footnote{Interpreted as a linear functional $D(i_\cap^{\mathrm{sx}}[f])$ acting on $g$, this expression can be seen to be the Fréchet derivative of the local redundancy fulfilling~\cite{cartan1971differential}
    \begin{equation*}
        \lim_{||\epsilon g|| \to 0} \frac{||i^\mathrm{sx}_\cap[f+\epsilon g] - i^\mathrm{sx}_\cap[f] - D(i_\cap^{\mathrm{sx}}[f])g||}{||\epsilon g||} = 0 \, .
    \end{equation*}
    }

    Considering the form that results from the derivative, it can be concluded that the $i_\cap^{\mathrm{sx}}[f_{TS_1S_2}](t:s_1 \lor s_2)$ is not only differentiable at point $(t, s_1, s_2)$, but even smooth, since the following derivatives will be linear combinations of powers of terms of the form $\frac{1}{f_{TS_1} + f_{TS_2}}$, $\frac{1}{f_{S_1} + f_{S_2}}$ and $\frac{1}{f_{T}}$. 
    
    Note further that this proof naturally extends to any arbitrary antichain $\alpha$ of interest by an analogous argument.
    \end{proof}

    \begin{corollary}
        The global measure of redundant information $I^{\mathrm{sx}}_\cap[f] = \int f \, i^{\mathrm{sx}}_\cap[f]$ is a smooth functional of the underlying density $f$ if $f$ and its marginals are bounded from below by some $\xi > 0$.
    \end{corollary}
    \begin{proof}
    The change of $I^\mathrm{sx}_\cap$ when shifting $f$ infinitesimally in the direction of $g$ is given by
    \begin{align*}
            \left[\frac{\mathrm{d}I^\mathrm{sx}_\cap[f+ \epsilon g]}{\mathrm{d}\epsilon}\right]_{\epsilon = 0} = \int \mathrm{d}t\mathrm{d}s_1\mathrm{d}s_2 \, \frac{\delta I^\mathrm{sx}_\cap[f]}{\delta f(t, s_1, s_2)} g(t, s_1, s_2)  \, .
    \end{align*}
    where $\delta I^\mathrm{sx}_\cap[f]/\delta f(t, s_1, s_2)$ is the functional derivative of $I^\mathrm{sx}_\cap$~\cite{cartan1971differential}.
    By applying a product rule of differentiation to the functional
    \begin{align*}
        I^\mathrm{sx}_\cap[f] = \int \mathrm{d}t\,\mathrm{d}s_1\,\mathrm{d}s_2\,
        f(t, s_1, s_2) \, i^\mathrm{sx}_\cap[f](t, s_1, s_2) \, ,
    \end{align*}
    which depends on $f$ in two places, we obtain
    \begin{gather*}
            \left[\frac{\mathrm{d}I^\mathrm{sx}_\cap[f+ \epsilon g]}{\mathrm{d}\epsilon}\right]_{\epsilon = 0}\\
            = \left[\frac{\mathrm{d}}{\mathrm{d}\epsilon}\int \mathrm{d}t\,\mathrm{d}s_1\,\mathrm{d}s_2\, (f+\epsilon g)(t, s_1, s_2) \right.\\
            \left.\vphantom{\int} i^\mathrm{sx}_\cap[f+\epsilon g](t, s_1, s_2)\right]_{\epsilon = 0}\\
    \end{gather*}
    \begin{gather*}
            = \int \mathrm{d}t\mathrm{d}s_1\mathrm{d}s_2 \, \{i^\mathrm{sx}_\cap[f](t, s_1, s_2) \, g(t, s_1, s_2)\\
            + f(t, s_1, s_2) \left[D(i^\mathrm{sx}_\cap[f])g\right](t, s_1, s_2)\}\\
            = \int \mathrm{d}t\mathrm{d}s_1\mathrm{d}s_2 \, \left[\left\{i^\mathrm{sx}_\cap[f] + f D(i^\mathrm{sx}_\cap[f])\right\}g\right](t, s_1, s_2)\,,
    \end{gather*}
    from which we can read off the functional derivative as the functional

    \begin{align*}
        \frac{\delta I^\mathrm{sx}_\cap[f]}{\delta f(t, s_1, s_2)}
    = i^\mathrm{sx}_\cap[f] + f D(i^\mathrm{sx}_\cap[f]) \,.
    \end{align*}
    If the marginals of the joint density $f$ are bounded from below, the magnitude of the Fréchet derivative, $|D_{g}(i_\cap^{\mathrm{sx}}[f])|$, is bounded, too. Hence the functional derivative is well-defined for all points $(t, s_1, s_2)$.
    Thus the total change of the functional $I^\mathrm{sx}_\cap[f]$ with respect to an infinitesimal change of $f$ in the direction of $g$ exists if and only if both $i^\mathrm{sx}_\cap[f]g$ and $f D(i^\mathrm{sx}_\cap[f])g$ are integrable functions over the space of realizations.
    \end{proof}

\section{Analytical probability distributions for the logic gates and sum examples}
\label{apx:examples_distributions}

The probability densities for the three continuous logic gates as well as the sum example can be expressed as multivariate Gaussian distributions of the form

\begin{equation*}
    p_{T, S_1, S_2}(t, s_1, s_2) = \frac{1}{\sqrt{(2 \pi)^l \det \Sigma}} \exp{-\frac{1}{2} 
\bm{x}^T \, \Sigma^{-1} \, \bm{x}} \, ,
\end{equation*}
where $\bm{x} = \begin{pmatrix}t&s_1&s_2\end{pmatrix}^T$ and the covariance matrix

\begin{equation*}
    \Sigma = 
    \begin{pmatrix}
    \Sigma_T & \Sigma_{TS}\\
    \Sigma_{ST} & \Sigma_S
    \end{pmatrix}
\end{equation*} depends on the gate. Note, further, that the integer $l$ is equal to $3$ for the redundant and unique gate as well as the sum example while it is equal to $4$ for the copy gate because of the two-dimensional target.

An analytical formulation for the mutual information for such Gaussian models is given by \cite{coverthomas}

\begin{equation*}
    I(T, \bm S) = \frac{1}{2} \log_2 \left(\frac{\det \Sigma_T \det \Sigma_S}{\det \Sigma}\right).
\end{equation*}
To make the mutual information finite, noise with standard deviation $\sigma = 0.01$ is added to the target variable. While the choice of this parameter influences the absolute sizes of the PID atoms, it does not majorly affect the qualitative analysis of the logic gates (\cref{tab:sigma_gates}). However, analogously to the KSG estimator for mutual information, a higher number of samples is necessary to correctly estimate the redundant information when the dependence between the variables is strong~\cite{kraskov}.
Furthermore, the constant $\delta = 10^{-9}$ is introduced in the redundant gate for numerical reasons to avoid a singular matrix. It has been carefully chosen small enough to not affect the results up to the presented precision.

\begin{table*}[!]
    \centering
    \caption{\textbf{The PID of the four toy examples does not qualitatively differ when varying the amount of noise $\sigma$ that is added to the source variable $T$.} The results have been computed using a Monte-Carlo integration analogous to \cref{tab:analytical_examples}.}
    \begin{ruledtabular}
    \begin{tabular}{rcccc}
        &  Redundant Gate & Copy Gate & Unique Gate & Sum\\
        \hline
        $\sigma = 0.01$\\
        \hline
        $\Pi_\mathrm{syn}$ & $0.000 \, \mathrm{bits}$ & $6.644 \, \mathrm{bits}$ & $5.525  \, \mathrm{bits}$ & $6.497  \, \mathrm{bits}$\\
        $\Pi_\mathrm{unq,1}$ & $0.000    \, \mathrm{bits}$ & $0.000 \, \mathrm{bits}$ & $1.119 \, \mathrm{bits}$ & $0.147  \, \mathrm{bits}$\\
        $\Pi_\mathrm{unq,2}$ & $0.000    \, \mathrm{bits}$ & $0.000 \, \mathrm{bits}$ & $-5.525 \, \mathrm{bits}$ & $0.147  \, \mathrm{bits}$\\
        $\Pi_\mathrm{red}$ & $6.644 \, \mathrm{bits}$ & $6.644 \, \mathrm{bits}$ & $5.525   \, \mathrm{bits}$ & $0.353  \, \mathrm{bits}$\\
        $I(T:S_1, S_2)$ & $6.644 \, \mathrm{bits}$ &  $13.288 \, \mathrm{bits}$ & $6.644 \, \mathrm{bits}$ & $7.144  \, \mathrm{bits}$\\
        \hline
        $\sigma = 0.1$\\
        \hline
        $\Pi_\mathrm{syn}$ & $0.000 \, \mathrm{bits}$ & $3.322 \, \mathrm{bits}$ & $2.388  \, \mathrm{bits}$ & $3.186  \, \mathrm{bits}$\\
        $\Pi_\mathrm{unq,1}$ & $0.000 \, \mathrm{bits}$ & $0.000 \, \mathrm{bits}$ & $0.942 \, \mathrm{bits}$ & $0.143  \, \mathrm{bits}$\\
        $\Pi_\mathrm{unq,2}$ & $0.000 \, \mathrm{bits}$ & $0.000 \, \mathrm{bits}$ & $-2.388 \, \mathrm{bits}$ & $0.143  \, \mathrm{bits}$\\
        $\Pi_\mathrm{red}$ & $ 3.329 \, \mathrm{bits}$ & $3.329 \, \mathrm{bits}$ & $2.388 \, \mathrm{bits}$ & $0.353  \, \mathrm{bits}$\\
        $I(T:S_1, S_2)$ & $ 3.329 \, \mathrm{bits}$ &  $6.658 \, \mathrm{bits}$ & $3.329 \, \mathrm{bits}$ & $3.826  \, \mathrm{bits}$\\
        \hline
        $\sigma = 0.001$\\
        \hline
        $\Pi_\mathrm{syn}$ & $0.000 \, \mathrm{bits}$ & $9.966 \, \mathrm{bits}$ & $8.822 \, \mathrm{bits}$ & $9.819  \, \mathrm{bits}$\\
        $\Pi_\mathrm{unq,1}$ & $0.000 \, \mathrm{bits}$ & $0.000 \, \mathrm{bits}$ & $1.143 \, \mathrm{bits}$ & $0.147 \, \mathrm{bits}$\\
        $\Pi_\mathrm{unq,2}$ & $0.000 \, \mathrm{bits}$ & $0.000 \, \mathrm{bits}$ & $-8.822 \, \mathrm{bits}$ & $0.147 \, \mathrm{bits}$\\
        $\Pi_\mathrm{red}$ & $9.965 \, \mathrm{bits}$ & $9.966 \, \mathrm{bits}$ & $8.822\, \mathrm{bits}$ & $0.353  \, \mathrm{bits}$\\
        $I(T:S_1, S_2)$ & $9.966 \, \mathrm{bits}$ &  $19.932 \, \mathrm{bits}$ & $9.966 \, \mathrm{bits}$ & $10.466  \, \mathrm{bits}$\\
    \end{tabular}
    \end{ruledtabular}

    \label{tab:sigma_gates}
\end{table*}

The covariance matrices are found as follows:

\paragraph{Redundant gate}
\begin{equation*}
    \Sigma =
    \begin{pmatrix}
        1 + \sigma^2 & 1 - \delta & 1 - \delta\\
        1 - \delta & 1 & 1 - \delta\\
        1 - \delta & 1 - \delta & 1
    \end{pmatrix}
\end{equation*}

\paragraph{Copy gate}
\begin{equation*}
    \Sigma = 
    \begin{pmatrix}
        1 + \sigma^2 & 0 & 1 & 0\\
        0 & 1 + \sigma^2 & 0 & 1\\
        1 & 0 & 1 & 0\\
        0 & 1 & 0 & 1
    \end{pmatrix}
\end{equation*}
where $l=4$ and $\bm x = \begin{pmatrix}t_1, t_2, s_1, s_2\end{pmatrix}^T$,

\paragraph{Unique gate}
\begin{equation*}
   \Sigma = 
   \begin{pmatrix}
       1 + \sigma^2 & 1 & 0\\
       0 & 1 & 0\\
       1 & 0 & 1\\
   \end{pmatrix}
\end{equation*}

\paragraph{Continuous sum}
\begin{equation*}
    \Sigma =
    \begin{pmatrix}
        2 + \sigma^2 & 1 & 1\\
        1 & 1 & 0\\
        1 & 0 & 1
    \end{pmatrix} \, .
\end{equation*}

\section{Formal derivation of the Kozachenko-Leonenko estimator for Shannon differential entropy}
\label{apx:kl}

The proposed estimator for continuous $I^\mathrm{sx}_\cap$ is based on ideas of the KSG estimator, originally introduced by \citet{kraskov} in 2004, which is explained here for completeness. The main advancement of the KSG estimator was to apply the Kosachenko-Leonenko estimator \cite{kl} for (differential) entropy in a specific way, allowing for estimation of $I(X;Y) = H(X) + H(Y) - H(X,Y)$. Here we will continue in a very similar fashion. 

Kosachenko and Leonenko estimate (differential) entropy on the grounds of the assumption that the density to consider for the entropy is locally constant around a measured data point $x_i$, i.e., $q_i(\epsilon) = \int\limits_{B_\epsilon(x_i)} \mathrm{d}x \, p_{X}(x) \approx \lambda_X\left( B_\epsilon(x_i) \right) p_X(x_i)$, where $B_\epsilon(x_i)$ is the $\epsilon$-ball around $x_i$. Then Kosachenko and Leonenko introduce a probability density function $P_k(\epsilon)$, such that $P_k(\epsilon) d\epsilon$ determines the probability of  $k-1$ points lying within radius $\epsilon$ around $x_i$, a $k$-th lying within the shell of radii within $[\epsilon, \epsilon + d\epsilon]$ around $x_i$, and the remaining $N-k-1$ points lying outside.

Using well-known integral identities, we find
\begin{align*}
    \int\limits_{\mathbb{R}^+} \mathrm{d}\epsilon \, \log[q_i(\epsilon)] P_k(\epsilon) = \psi(k) - \psi(N) \, .
\end{align*}

On the other hand, since $p_X(x_i) \lambda_X\left( B_{\epsilon_i}(x_i) \right) \approx q_i(\epsilon_i)$, it follows that
\begin{gather*}
    \log\left[p_{X}(x_i) \lambda_X\left( B_{\epsilon_i}(x_i) \right)\right] \approx \mathbb{E}\left[\log[q_i]\right] = \psi(k) - \psi(N)
\end{gather*}
such that a point estimate $\hat H$ of the differential entropy can be determined as
\begin{gather*}
    \hat{H}(X) = -\frac{1}{N} \sum\limits_{i=1}^N \widehat{\log\left[p_{X}(x_i)\right]}\\
    = - \psi(k) + \psi(N) + \frac{1}{N} \sum\limits_{i=1}^N \log\left[\lambda_X\left( B_{\epsilon_i}(x_i)\right)\right].
\end{gather*}

The volume bias terms $\lambda_X\left( B_\epsilon(x_i)\right)$ are numerically obtained by performing searches for the $k$-th nearest neighbor in the data set, the distance to whom is $\epsilon_i$.

\section{Estimation steps} \label{appendix_b_estimation_steps}

In \cref{subsection:generalizing_kraskov} we estimate $\log[f_\alpha(t, s_i)]$ for $\alpha = \bigvee_j a_j$ containing an OR statement by
\begin{gather*}
    \widehat{\log[f^i_\alpha(t, s_k)]} = \widehat{\log\left[ \sum_j f^i_{a_j}(t, s_k)\right]} \\
    \approx \log\left[ \sum _j e^{\widehat{\log[f^i_{a_j}]}}\right]\\
\end{gather*}
\begin{gather*}
    \overset{\text{KSG}}{=} -\psi(N) + \psi(k) + \log\left[ \sum_j v_{a_j, \epsilon}^i \right]
\end{gather*}
where the $v_{a_j, \epsilon}^i$ each represent the volume terms in the estimation at the $i-$th point. To the end of arriving at a form which is treatable by considerations akin to the KSG estimator, we have intermediately used an expansion of the sort $z = \exp[\log[z]], \, z \geq 0$.
Then, since by usage of the maximum norm $v_{a_j, T, \epsilon}^i = v_{a_j, \epsilon}^i \cdot v_{T, \epsilon}^i$, we find
\begin{equation}
\begin{gathered}
    \hat{I}^{sx}_\cap = \langle\widehat{\log[f^i_{\alpha, T}(t, s_k)]}\rangle_i - \langle\widehat{\log[f^i_\alpha(t, s_k)]}\rangle_i - \langle\widehat{\log[f^i_T(t, s_k)]}\rangle_i \nonumber \\
    = \psi(k) + \psi(N) - \langle \psi(n_\alpha(i))\rangle_i - \langle \psi(n_T(i))\rangle_i  
\end{gathered}
\end{equation}

\section{Implementation details}
\label{apx:implementation}

Having introduced an estimator for continuous $I^\mathrm{sx}_\cap$ in Sections \ref{sec:estimator} and \ref{sec:estimator_multivariate}, this section explains how this estimator can be realized in an efficient algorithm.

In the estimation procedure, the most computationally expensive step is the search for $k$ nearest neighbors for each sample point in the joint space and counting of neighbours within a ball in the marginal spaces. For the maximum norm, efficient $\mathcal O(n \log n)$ algorithms can be found for both low- ($kd$-trees) and high-dimensional (ball tree) data. Furthermore, approximate procedures exist for scaling beyond what the exact algorithms can compute.

In order to determine the distance $\epsilon$ to the $k$th nearest neighbor from a given disjunction statement, we proceed as follows: First, we obtain the $k$ nearest neighbors in each of the sets of variables of the corresponding antichain. Then we do a merging procedure: To get the $k$ closest points to any of the source variables, we successively take points with the smallest distance in any one of the subspaces until we have a total of $k$ points. While the kd-tree algorithm could be adapted to directly find neighbours according to this custom distance function, using the merging-procedure we achieve the same runtime complexity and can use highly optimized existing kd-tree implementations. The pseudocode for this function is given in \cref{alg:compute_epsilons}.

Following an analogous merging procedure, the number of points within a radius of $\epsilon$ from the disjunction statement in the source space (see \cref{alg:compute_n_alpha}) as well as the number of points within a radius $\epsilon$ from the query point in the target space (see \cref{alg:compute_n_T}) is determined. From these results, the redundancy can readily be computed according to \cref{alg:compute_redundancy}.

We implemented the described algorithm as a python package using SciPy~\cite{2020SciPy-NMeth} for efficient nearest-neighbor searches. The package is publicly available under 

\begin{center}
\href{https://gitlab.gwdg.de/wibral/continuouspidestimator}{gitlab.gwdg.de/wibral/continuouspidestimator}
\end{center}   

\begin{figure*}
\begin{minipage}[t]{.95\columnwidth}
\begin{algorithm}[H]
    \KwData{Source samples $(\bm S_{j})_j$; Target samples $(T_{j})_j$; Antichain $\alpha \in \mathcal{P}(\mathcal{P}(\{1, \dots, n\}$))}
    \KwResult{Distances from the disjunction to the $k$-th nearest neighbor in the joint source-target space.}
    
    \tcp{For each sample point indexed by $j$}

    \For{$j \in (1, \dots, n_{\mathrm{samples}})$}{

        \tcp{Find the $k$ nearest neighbors in each of the spaces $\bm S_{\bm a} \cross T$ for all sets $\bm a$ in the antichain $\alpha$}
        
        \For{$\bm a \in \alpha$}{

            $\mathrm{kNN}^{\bm ST}_{\bm aj} \gets \mathrm{find\_k\_nearest\_neighbors}(\bm{S}_{\bm aj} \cross T_j)$

        }
    
        \tcp{Merge points from all separate nearest neighbor searches in the individual spaces to get the $k$ unique samples with the smallest distances to their respective query point $j$}
        $\mathrm{kNN\_merged^{\bm ST}_{j}} \gets \mathrm{merge}(\left\{\mathrm{kNN}^{\bm ST}_{\bm aj}\mid|\bm a \in \alpha \right\})$

        \tcp{Get search radius $\epsilon$ as the kth smallest distance}
        $\epsilon_j \gets \mathrm{dist\_to\_kth}(\mathrm{kNN\_merged^{\bm ST}_{j}})$
    }

    \Return{$(\epsilon_j)_j$}
    \caption{compute\_epsilons}
    \label{alg:compute_epsilons}
\end{algorithm}
\vspace{2cm}
\end{minipage}
\begin{minipage}[t]{0.95\columnwidth}
\begin{algorithm}[H]
    \KwData{Source samples $(\bm S_{j})_j$;  Antichain $\alpha \in \mathcal{P}(\mathcal{P}(\{1, \dots, n\}$)); Distances from the disjunction to the $k$-th nearest neighbor in the joint source-target space $(\epsilon_j)_j$}
    \KwResult{Number of points in the source disjunction space with distance less than $\epsilon_j$ from the disjunction}

    \For{$j \in (1, \dots, n_{\mathrm{samples}})$}{

        \tcp{Get nearest neighbors in ball of radius $\epsilon$ in the individual conjunction spaces}

        \For{$\bm a \in \alpha$}{
            $\mathrm{NN}^{\bm S}_{\bm aj} \gets \mathrm{find\_points\_within\_distance}(\bm S_{\bm aj}, \epsilon_j)$
        }
    
        \tcp{Count unique samples within any of the epsilon balls for the individual conjunction spaces}
        
        $n_{\alpha j}$ $\gets \mathrm{count\_unique}(\left\{\mathrm{NN}^{\bm S}_{\bm aj}\mid|\bm a \in \alpha \right\})$\\
    }
    \Return{$(n_{\alpha j})_j$}
    \caption{compute\_$n_\alpha$}
    \label{alg:compute_n_alpha}

\end{algorithm}
\end{minipage}
\\
\begin{minipage}[t]{.95\columnwidth}
\begin{algorithm}[H]
    \KwData{Target samples $(T_{j})_j$; Antichain $\alpha \in \mathcal{P}(\mathcal{P}(\{1, \dots, n\}$)); Distances from the disjunction to the $k$-th nearest neighbor in the joint source-target space $(\epsilon_j)_j$}
    \KwResult{Number of points in the target space with distance less than $\epsilon_j$ from the disjunction}

    \For{$j \in (1, \dots, n_{\mathrm{samples}})$}{

        \tcp{Get nearest neighbors in ball of radius $\epsilon$ in the target space}

        $\mathrm{NN}^{T}_{\bm j} \gets \mathrm{find\_points\_within\_distance}(\bm T_{j}, \epsilon_j)$
    
        \tcp{Count samples within the epsilon ball for the target space}
        
        $n_{T j}$ $\gets \mathrm{count}(\mathrm{NN}^{T}_{j})$\\
    }
    \Return{$(n_{T j})_j$}

    \caption{compute\_$n_T$}
    \label{alg:compute_n_T}
\end{algorithm}
\end{minipage}
\begin{minipage}[t]{.95\columnwidth}
\begin{algorithm}[H]
    \KwData{Source samples $(\bm S_{j})_j$; Target samples $(T_{j})_j$; Antichain $\alpha \in \mathcal{P}(\mathcal{P}(\{1, \dots, n\}$))}
    \KwResult{Distances from the disjunction to the $k$-th nearest neighbor in the joint source-target space.}

    \tcp{Compute distances in joint space}
    $\epsilon \gets \mathrm{compute\_epsilons}((\bm S_j)_j, (T_j)_j)$

    \tcp{Compute number of neighbors in source space}
    $n_\alpha \gets \mathrm{compute\_}n_\alpha((\bm S_j)_j, \epsilon)$

    \tcp{Compute number of neighbors in target space}
    $n_T \gets \mathrm{compute\_}n_\alpha((T_j)_j, \epsilon)$
    
    \tcp{Compute estimate for redundancy}
    $I_\cap \gets \Psi(k) + \Psi(n_\mathrm{samples}) - \left<\Psi(n_\alpha[j]) + \Psi(n_T[j])\right>_j$
    
    \Return{$I_\cap$}\;
    \caption{compute\_redundancy}
    \label{alg:compute_redundancy}
\end{algorithm}
\end{minipage}
\end{figure*}

\section{Role of the parameter $k$}
\label{apx:k}

The characteristics of the estimates produced by the procedure laid out in \cref{sec:estimator} are influenced by the choice of the parameter $k$, which determines which neighbors are used to compute the distance $\epsilon$ in the joint space by requiring that the $k-1$ closest neighbors are skipped. For larger $k$, these distances $\epsilon$ will thus be larger as well, which has two important ramifications: Firstly, since the considered $\epsilon$-neighborhoods are larger and contain more points, fluctuations in the positions of individual points have less impact on the estimation leading to less variance. Secondly, however, since the estimator relies on the assumption that the probability density is approximately constant over the $\epsilon$-neighborhood, larger $k$ for a fixed $N$ may make the estimator unable to resolve high-frequency fluctuations which occur in the data. This may lead to a higher bias in the estimation for a given number of samples $N$.
The estimation results do not depend on the parameter $k$ alone but only on the fraction $k/N$ (as visible in \cref{fig:comparison_k} and \cref{fig:comparison_kN}), as this determines the average distance $\epsilon$.

\vspace{-\baselineskip}

\begin{figure*}[ht]
    \begin{minipage}{0.45\textwidth}
    \centering
    \includegraphics[width=\columnwidth]{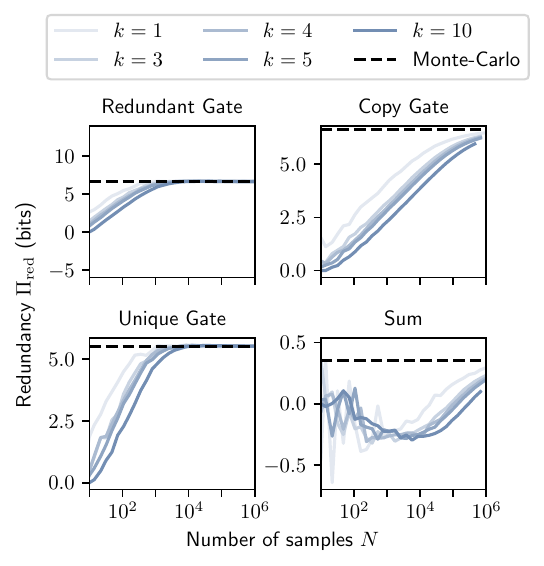}
    \caption{\textbf{On the demo gates, the estimation converges to the Monte-Carlo results with increasing number of samples $N$ for all fixed parameters $k$}}
    \label{fig:comparison_k}
    \end{minipage}\qquad
    \begin{minipage}{0.45\textwidth}
    \centering
    \includegraphics[width=\columnwidth]{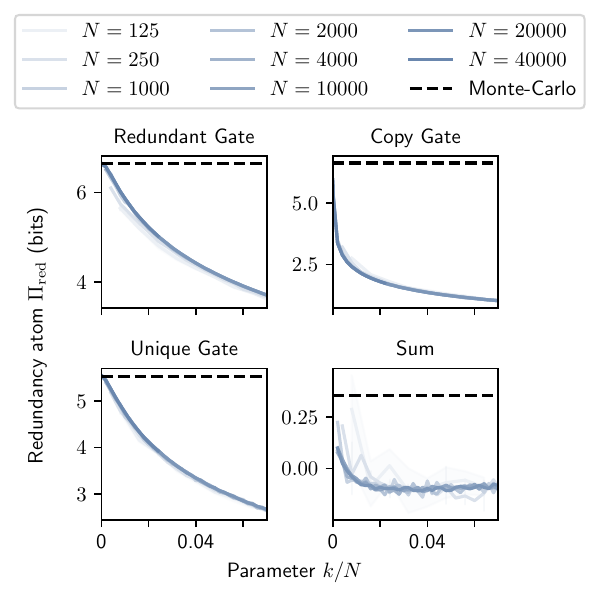}
    \caption{\textbf{The result of the estimation does not depend on the value of the parameter $k$ alone but only on its fraction $k/N$.}}
    \label{fig:comparison_kN}
    \end{minipage}
\end{figure*}

\section{Skewed sum} \label{apx:skewed_sum}
We define the covariance matrix for the skewed sum as
\begin{equation*}
    \Sigma =
    \begin{pmatrix}
        2.2 + \sigma^2 & 1.2 & 1\\
        1.2 & 1.2 & 0\\
        1 & 0 & 1
    \end{pmatrix} \, .
\end{equation*}

The results of the continuous PID analysis are summarized in table \cref{tab:skewed_sum}

\begin{table}[!]
    \centering
    \caption{\textbf{PID results for skewed sum example} The results have been computed using a Monte-Carlo integration analogous to \cref{tab:analytical_examples}.}
    \begin{ruledtabular}
    \begin{tabular}{rc}
        & Skewed Sum\\
        \hline
        $\Pi_\mathrm{syn}$ & $9.884 \, \mathrm{bits}$\\
        $\Pi_\mathrm{unq,1}$ & $0.214 \, \mathrm{bits}$\\
        $\Pi_\mathrm{unq,2}$ & $0.082 \, \mathrm{bits}$\\
        $\Pi_\mathrm{red}$ & $0.355 \, \mathrm{bits}$\\
        $I(T:S_1, S_2)$ & $ 10.535 \, \mathrm{bits}$
        
    \end{tabular}
    \end{ruledtabular}
    \label{tab:skewed_sum}
\end{table}

\section{A Glimpse into Mixed systems of discrete and continuous variables}
\label{apx:mixed}

Although systems including both discrete and continuous variables do at this point in time exceed the capabilities of the estimator developed in this work, we still intend to share some theoretical aspects about potential future endeavours.
Thus, we propose the following treatment of systems, which are composed of both discrete and continuous variables. Assuming a mixture of discrete and continuous sources $S_j$ and an either discrete of continuous target $T$, we suggest utilizing the same exponential expansion as in the purely continuous case \cref{definition_1}. This leads to the approximation
\begin{widetext}
\begin{align*}
    \widehat{\log[f^i_\alpha]} \approx \log\left[ \sum\limits_{k=1}^{|I_d| + 1} (-1)^{k+1} \left(\sum\limits_{\substack{m_1 < m_2 < \ldots < m_{k-1} \in I_d \\ m_1 \neq m_2 \neq \ldots \neq m_{k-1} \neq j \in J}} e^{\widehat{\log\left[f_{\bigcup\limits_{q = 1}^{k-1} a_{m_q} \cup a_j}\left(\bigcup\limits_{q = 1}^{k-1} x_{a_q} \cup x_{a_j}\right)\right]}} \right) \right] \, .
\end{align*}
\end{widetext}
Here then we have a number of the possible combinations of statements $\bigcup\limits_{q = 1}^{k-1} a_{m_q} \cup a_j$ containing at least one discrete variable (denoted by $a_j$). The statements consisting solely of continuous variables will be treated individually by KSG-like estimation while the others will be handled by isolating the discrete variables and conditioning on those, using the identity $H(\{S_k\}_{k \in K}) - H( \{S_m\}_{m \in K \setminus L}) = H(\{S_l\}_{l \in L}| \{S_n\}_{n \in K \setminus L})$ for some index sets $L \subset K$. This leads to
\begin{widetext}
\begin{align*}
    &\log\left[f^i_{\bigcup\limits_{q = 1}^{k-1} a_{m_q} \cup a_j}\right] = \log\left[f^i_{c|d (\bigcup\limits_{q = 1}^{k-1} a_{m_q} \cup a_j)}\right] + \log\left[f^i_{disc(\bigcup\limits_{q = 1}^{k-1} a_{m_q} \cup a_j)}\right]\\
    \implies &\widehat{\log\left[f^i_{\bigcup\limits_{q = 1}^{k-1} a_{m_q} \cup a_j}\right]} = \widehat{\log\left[f^i_{c|d (\bigcup\limits_{q = 1}^{k-1} a_{m_q} \cup a_j)}\right]} + \widehat{\log\left[f^i_{disc(\bigcup\limits_{q = 1}^{k-1} a_{m_q} \cup a_j)}\right]} \\
    &\,\overset{\text{KSG}}{=} \psi(k) - \psi(N) - \log\left[v_{\epsilon, c|d(\bigcup\limits_{q = 1}^{k-1} a_{m_q} \cup a_j)}^{i}\right] +  \widehat{\log\left[f^i_{disc(\bigcup\limits_{q = 1}^{k-1} a_{m_q} \cup a_j)}\right]} \, ,
\end{align*}
\end{widetext}
where we have omitted the explicit dependent variables $\bigcup\limits_{q = 1}^{k-1} x_{a_q} \cup x_{a_j}$ for simplicity. Further, $c|d(a)$ denotes the continuous random variables in $a$ conditioned on the discrete variables in $a$.
Here we have treated the first estimated logarithm via KSG estimation, as here the only dependent variables are purely continuous, while the latter term can be evaluated via a standard plugin estimator as it is a purely discrete entropy.
Thus the estimators for $\alpha$ and $(\alpha, T)$ become
\begin{widetext}
\begin{align*}
    \widehat{\log[f^i_\alpha]} &= \psi(k) - \psi(N) -  \underbrace{\left\langle \log\left[\sum\limits_{k=1}^{|I_d| + 1} (-1)^{k+1} \left(\sum\limits_{\substack{m_1 < m_2 < \ldots < m_{k-1} \in I_d \\ m_1 \neq m_2 \neq \ldots \neq m_{k-1} \neq j \in J}} v_{\epsilon, c|d(\bigcup\limits_{q = 1}^{k-1} a_{m_q} \cup a_j)}^{i} \hat{f}^i_{disc(\bigcup\limits_{q = 1}^{k-1} a_{m_q} \cup a_j)} \right) \right]  \right\rangle_i}_{=:G_\alpha} \text{ and}
\end{align*}
\begin{align*}
    \widehat{\log[f^i_{\alpha, T}]} &= \psi(k) - \psi(N) - \underbrace{\left\langle \log\left[\sum\limits_{k=1}^{|I_d| + 1} (-1)^{k+1} \left(\sum\limits_{\substack{m_1 < m_2 < \ldots < m_{k-1} \in I_d \\ m_1 \neq m_2 \neq \ldots \neq m_{k-1} \neq j \in J}} v_{\epsilon, c|d(\bigcup\limits_{q = 1}^{k-1} a_{m_q} \cup a_j, \, T)}^{i} \hat{f}^i_{disc(\bigcup\limits_{q = 1}^{k-1} a_{m_q} \cup a_j, \, T)} \right) \right]  \right\rangle_i}_{=:G_{\alpha, T}} \, .
\end{align*}    
\end{widetext}
Using the joint space for initial determination of $\epsilon^i$ as the distance to $(x_1,\ldots x_n)^i$-th $k$-th nearest neighbor, the final estimator reads
\begin{widetext}
\begin{align*}
    \hat{I}^{\mathrm{sx}}_\cap =\begin{cases} \psi(k_{c|d(\alpha, T)}) + \psi(N) - \langle \psi(n_{c|d(\alpha)}(i))\rangle_i - \langle \psi(n_T(i))\rangle_i - \Delta \mathrm{vol}_{\alpha, T} & T \text{ cont.} \\
    \psi(k_{c|d(\alpha, T)}) + \psi(N) - \langle \psi(n_{c|d(\alpha)}(i))\rangle_i & T \text{ disc.}\end{cases} \, ,
\end{align*}
\end{widetext}
with $\Delta \mathrm{vol}_{\alpha, T} = G_{\alpha, T} - G_{\alpha} - \langle\log[v^i_{\epsilon, T}]\rangle_i$ the resulting differential volume term, $c|d(\alpha)$ being the continuous statement, each conditioned on the discrete variables. 

\begin{example}
Considering $\alpha = X_1 \land X_2 \lor X_3 \lor X_4 \land X_5$, with $X_1, X_2$ and $X_4$ discrete, $X_3, X_5$ continuous, we find $c|d(\alpha)= X_3 \lor X_5|X_4$. Note $X_1 \lor X_2$ vanished from the continuous statement as they have already been treated inside $G_{\alpha, T}$ and $G_{\alpha}$.
\end{example}

\onecolumngrid
\clearpage
\twocolumngrid
\bibliographystyle{unsrtnat}
\bibliography{bibliography}

\end{document}